%% file: main_v3.tex
\documentclass[aps,prl,reprint,superscriptaddress,nofootinbib]{revtex4-2}

\usepackage[T1]{fontenc}
\usepackage[utf8]{inputenc}
\usepackage{amsmath,amssymb,amsthm,mathtools}
\usepackage{braket}
\usepackage{bm}
\usepackage{hyperref}
\usepackage{cleveref}
\usepackage{enumitem}
\usepackage{xcolor}
\usepackage{graphicx}
\usepackage{tikz}
\usepackage{quantikz}
\usepackage{orcidlink}
\usepackage{appendix}

\graphicspath{{./images/},{./imagesAppendix/}}

\definecolor{THc}{rgb}{0.9,0.3,0.2}

\hypersetup{colorlinks=true,linkcolor=blue,citecolor=blue,urlcolor=blue}

\newtheorem{theorem}{Theorem}
\newtheorem{proposition}{Proposition}
\newtheorem{lemma}{Lemma}

\newtheorem{definition}{Definition}

\newcommand{\ii}{\mathrm{i}}
\newcommand{\Tr}{\operatorname{tr}}

\newcommand{\E}{\mathbb{E}}
\newcommand{\cG}{\mathcal{G}}

\newcommand{\FAF}{\operatorname{FAF}}

\newcommand{\Dtr}{D_{\operatorname{tr}}}
\newcommand{\id}{\mathbb{I}}

\usepackage[normalem]{ulem}

\newcommand{\idg}[1]{{\bfseries #1)}}

\newcommand{\prlsection}[1]{{\em {#1}.---}}

\newcommand{\SM}{SM}

\begin{document}

\title{Practical Tests and Witnesses of Fermionic non-Gaussianity}

\author{Tobias Haug~\orcidlink{0000-0003-2707-9962}}
\email{tobias.haug@u.nus.edu}
\affiliation{Quantum Research Center, Technology Innovation Institute, Abu Dhabi, UAE}

\author{Xhek Turkeshi~\orcidlink{0000-0003-1093-3771}}
\affiliation{Institut f\"ur Theoretische Physik, Universit\"at zu K\"oln, Z\"ulpicher Strasse 77, 50937 K\"oln, Germany}

\author{Piotr Sierant~\orcidlink{0000-0001-9219-7274}}
\affiliation{Barcelona Supercomputing Center Plaça Eusebi G\"uell, 1-3 08034, Barcelona, Spain}

\date{\today}

\begin{abstract}
Fermionic Gaussian states describe free fermions and underlie the mean-field picture of matter, from metals to superconductors; they are also efficiently simulable on classical computers. 
Departures from Gaussianity — the correlations produced by interactions — are therefore what make a fermionic system hard to simulate classically and useful for quantum computation, analogous to the role of magic in stabilizer-based quantum computation. Yet detecting and quantifying such non-Gaussianity at scale has remained challenging. 
Here we introduce practical tests and witnesses of fermionic non-Gaussianity built on fermionic antiflatness, a measure derived from the two-point covariance matrix.
We estimate it with two protocols — a two-copy Bell measurement and a single-copy scheme using commuting Majorana bilinears — that determine whether a state is Gaussian or far from it at lower measurement cost than existing approaches, using only operations native to fault-tolerant hardware. 
For mixed states, a purity-corrected witness certifies non-Gaussianity and remains robust under strong noise; running it on the IQM quantum processor, we find that noise can both reduce and enhance non-Gaussianity. 
Finally, we show that preparing pseudorandom fermionic states requires extensive non-Gaussianity. Together, these tools enable the study and certification of non-Gaussian fermionic resources on present-day quantum devices.
% Detecting when a quantum state leaves the efficiently simulable fermionic Gaussian regime is a central task for benchmarking quantum devices and certifying fermionic magic resources. 
% We develop practical tests and witnesses based on fermionic antiflatness (FAF), a covariance-matrix-based measure of non-Gaussianity. 
% For $n$-qubit states, we estimate FAF using two complementary protocols: two-copy Bell measurements and a single-copy scheme based on commuting matchings of Majorana bilinears.
% These yield testers that distinguish pure Gaussian states from states $\epsilon$-far from the Gaussian set, using $O(n^2/\epsilon^2)$ two-copy Bell measurements or $O(n^3/\epsilon^4)$ single-copy measurements, improving the state of the art in the dependence on both $n$ and $\epsilon$ for fixed-copy protocols.
% For mixed states, we introduce a purity-corrected FAF witness that certifies non-Gaussianity and is highly robust to noise. With our witness, we demonstrate on the IQM quantum computer that noise can both reduce and enhance non-Gaussianity. 
% Finally, by examining pseudo non-Gaussianity, we show that the cryptographic task of pseudorandom-state generation requires extensive fermionic non-Gaussianity.
% Together, these results provide experimentally accessible tools for detecting, witnessing, and quantifying non-Gaussian fermionic resources.
\end{abstract}

\maketitle

 \let\oldaddcontentsline\addcontentsline% Store \addcontentsline
\renewcommand{\addcontentsline}[3]{}% Make 

Fermionic Gaussian states~\cite{knill2001fermionic, Jozsa08, bravyi2004lagrangian} are a fundamental class of quantum states that are fully characterized by the two-fermion covariance matrix~\cite{Negele18quantum}, and hence can be efficiently simulated on a classical computer. 
They play a central role across quantum information, condensed matter~\cite{Peschel09, Amico08, Surace22} and high-energy physics~\cite{Sala18, Zohar15, Kelman24}.
In quantum information, they arise naturally as the states generated by
matchgate circuits~\cite{valiant2001quantum,Terhal02}, underlying simulation, certification and benchmarking protocols~\cite{Gluza18, Wan23,zhao2021fermionic, majsak2025simple, Heyraud2025, Sierant2026theory, Denzler24, Bittel2025pac}.
Supplementing matchgate circuits with beyond-Gaussian operations enables quantum advantage~\cite{oszmaniec2022fermion}, and non-Gaussian fermionic resources~\cite{hebenstreit2019all, hebenstreit2020computational, Lumia24, Gottlieb05, Reardonsmith24flo, Debertolis25natural} play a role analogous to magic resources in stabilizer computation~\cite{bravyi2005universal, veitch2014resource, Leone22sre}.

In condensed-matter theory, fermionic Gaussian states are the cornerstone of mean-field theory: the closest Gaussian state to a generic many-body wavefunction is the Hartree-Fock or Bardeen-Cooper-Schrieffer ansatz~\cite{FetterWalecka03,Negele18quantum,AltlandSimons10,deGennes99,Schrieffer18}, while in high-energy physics they underpin weak-coupling perturbation theory~\cite{Weinberg95}. In both cases non-Gaussianity arises from interactions. Capturing it in strongly correlated regimes has driven a hierarchy of non-Gaussian variational ans\"atze---Gutzwiller-projected wavefunctions~\cite{vollhardt1984almost, bunemann1998gutzwiller}, slave-boson saddle-points~\cite{kotliar1986slaveboson}, Jastrow-Slater states~\cite{capello2005variational,becca2017quantum}, dynamical mean-field theory~\cite{georges1996dynamical}, and fermionic non-Gaussian variational families~\cite{shi2018variational,hackl2020geometry}---all building non-Gaussian correlations on a Gaussian backbone.

A central problem across all these communities is how to scalably probe non-Gaussianity in a given fermionic state. This is part of a broader program of quantifying how much of a given resource is present in a quantum system~~\cite{buhrman2008quantum,montanaro2013survey}. While for coherence, entanglement, and magic a rich toolbox of low-cost tests and witnesses has been developed~~\cite{baumgratz2014quantifying,soleimanifar2022testing,gross2021schur,haug2026efficient,tarabunga2025quantifying,hinsche2024singlecopy}, progress for fermionic non-Gaussianity has been considerably more limited. Recent proposals include fermionic tomography~\cite{bittel2025optimal,mele2025efficient}, fermionic convolution~\cite{lyu2024convolution,coffman2025measuring} and random purifications~\cite{walter2025random}. A further step in this direction is fermionic antiflatness (FAF), a family of covariance-matrix-based quantities that vanish on pure Gaussian states and probe fermionic non-Gaussianity in many-body systems~\cite{faf2025}.

Here, we propose practical methods to test and witness fermionic magic resources (see Fig.~\ref{fig:sketch}).
We show that the FAF of $n$-qubit states can be estimated with 
%\old{$O(n^3)$ sampling complexity, either from} 
$O(n^2)$ sampling complexity from simple two-copy Bell measurements or from a single-copy scheme that groups commuting Majorana bilinears.
Building on this, we test whether a pure state is Gaussian or $\epsilon$-far from it in trace distance using 
%\old{$O(n^2/\epsilon^2)$} 
$O(n/\epsilon^2)$ two-copy Bell measurements or $O(n^3/\epsilon^4)$ single-copy measurements---the lowest two-copy and single-copy cost so far (Tab.~\ref{tab:testcomplexity})---using only Clifford operations suited to fault-tolerant hardware.
%Our tests offer the lowest proven cost so far (see Tab.~\ref{tab:testcomplexity}), and are based on Clifford operations suited for fault-tolerant quantum computers.
\begin{figure}[t!]
	\centering	
	\includegraphics[width=0.48\textwidth]{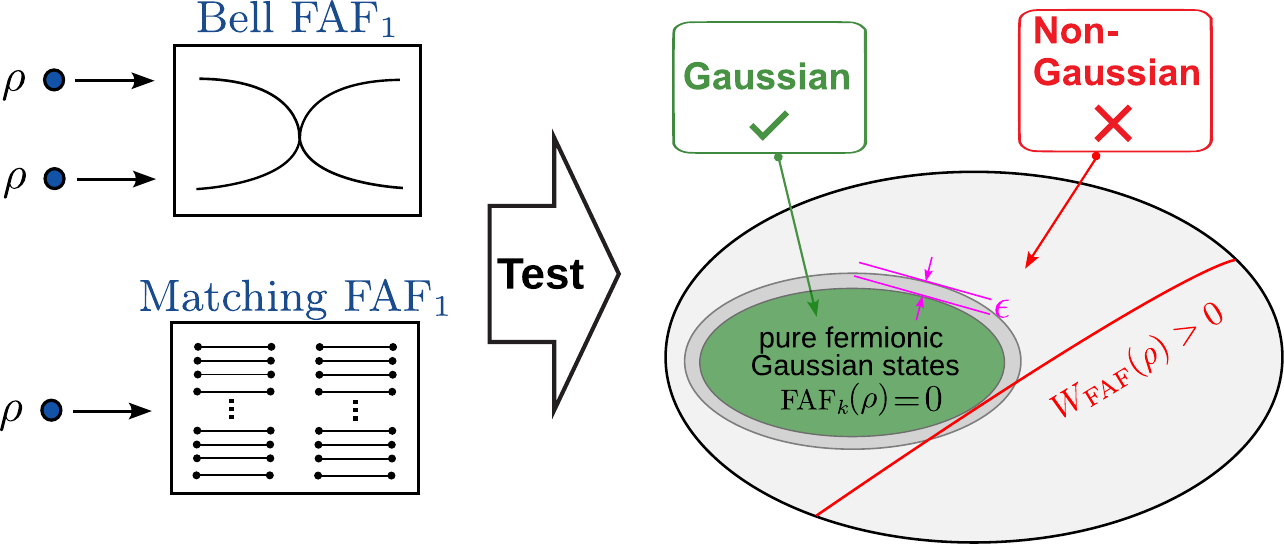}
	\caption{To test whether a given state $\rho$ is fermionic Gaussian or not, we use two-copy Bell measurements or single-copy scheme based on grouping matching bilinears to estimate fermionic anti-flatness (FAF).
	}
	\label{fig:sketch}
\end{figure}
For mixed states, a purity correction turns the FAF into a robust witness of non-Gaussianity that remains valid even under very strong depolarizing noise. We implement it on the IQM Garnet quantum computer, demonstrating that noise can both induce and suppress non-Gaussianity. Finally, we turn to quantum cryptography and show that pseudorandomness is not cheap in non-Gaussianity: pseudorandom states---indistinguishable from Haar-random ones by any polynomial-time quantum algorithm---exhibit near-maximal FAF and require $\Omega(n)$ local non-Gaussian gates.
%Our work demonstrates FAF as a practical tool to quantify the non-Gaussian resources of quantum devices.

\begin{table}[t!]
\centering
\setlength{\tabcolsep}{3.5pt}
\renewcommand{\arraystretch}{0.95}
\begin{tabular}{lcc}
\hline
Protocol & Samples & Copies \\
\hline
Bell $\FAF_1$ [this work]
&
%\old{$O(n^2/\epsilon^{2})$} 
$O(n/\epsilon^{2})$
&
2
\\
Matching $\FAF_1$ [this work]
&
$O(n^3/\epsilon^{4})$
&
1
\\
Random purification~\cite{walter2025random}
&
$O(n^2/\epsilon^{2})$
&
$O(n^2/\epsilon^2)$
\\
MG shadows/joint measurement~\cite{Wan23,zhao2021fermionic,majsak2025simple}
&
$O(n^4/\epsilon^{4})$
&
1
\\
Covariance spectrum~\cite{bittel2025optimal,mele2025efficient}
&
$O(n^5/\epsilon^{4})$
&
1
\\
Convolution~\cite{lyu2024convolution,coffman2025measuring}
&
\text{unknown}
&
3
\\
\hline
\end{tabular}
\caption{Pure-state Gaussianity testers: sample complexity and number of state copies in measurement round.  Here $\epsilon$ is the promised
trace-distance separation from the Gaussian set $\mathcal{G}_n$.  The displayed scalings keep only the leading polynomial
dependence on $n$ and $\epsilon^{-1}$; logarithmic factors in $n$ and
in the failure probability $\delta^{-1}$ are omitted. 
}
\label{tab:testcomplexity}
\end{table}

\prlsection{Quantifying fermionic non-Gaussianity}
We consider $n$ fermionic modes with $2n$ Majorana operators
$
\gamma_1,\ldots,\gamma_{2n}$, with anticommutation relation
$\{\gamma_a,\gamma_b\}=2\delta_{ab}\id$.
Jordan-Wigner mapping~\cite{Jordan1928} relates Majorana operators 
to the Pauli operators for the $j$th qubit $\{X_j, Y_j,Z_j\}$ via $\gamma_{2j-1}=\Big(\prod_{k<j} Z_k\Big)X_j$ and $\gamma_{2j}=
\Big(\prod_{k<j} Z_k\Big)Y_j $.
A (generally mixed) fermionic Gaussian state (FGS) takes the form 
\begin{equation}
\rho_G = \frac{\exp[-\tfrac{i}{2}\sum_{a<b}K_{ab}\gamma_a\gamma_b]}{\mathrm{tr}\left(\exp[-\tfrac{i}{2}\sum_{a<b}K_{ab}\gamma_a\gamma_b]\right)}
\label{eq:rhoG} 
\end{equation} 
with antisymmetric $K=-K^T\in\mathbb R^{2n\times 2n}$. 
The antisymmetric covariance matrix $\Gamma_\rho$ is given by
\begin{equation}
   (\Gamma_\rho)_{ab}=-i\frac{1}{2}\Tr(\rho[\gamma_a,\gamma_b]), 
\end{equation}
where singular values of $\Gamma_\rho$ are $\nu_j$, $j=1,\dots,n$ with each $\nu_j$ being twice degenerate and $0\le\nu_j\le1$. For Gaussian states, all higher-order Majorana correlation functions are
determined by \(\Gamma_{\rho_G}\) via Wick's theorem~\cite{Negele18quantum}.
Pure FGSs are characterized by
$\Gamma^2=-\id_{2n}$,
 or equivalently, $\nu_1=\cdots=\nu_n=1$.
We denote by $\cG_n$ the set of pure FGSs. 
Free fermionic unitaries are those generated by quadratic Majorana Hamiltonians,
\begin{equation}
    U_G=\exp\!\Big(-\frac{1}{4}\sum_{a,b}h_{ab}\gamma_a\gamma_b\Big),\quad h=-h^T\in\mathbb R^{2n\times 2n}.
\end{equation}

Now, to characterize Gaussianity, we consider the $k$th-order fermionic antiflatness (FAF)~\cite{faf2025} for $k\geq1$
\begin{equation}
\FAF_k(\rho)
=
n-\frac{1}{2}\Tr\!\left[(-\Gamma_\rho^2)^k\right]
=
n-\sum_{j=1}^n \nu_j^{2k}.
\label{eq:FAFdef}
\end{equation}
For pure states, \(\mathrm{FAF}_k=0\) iff the state is FGS, else \(\mathrm{FAF}_k>0\).
It is invariant under free fermionic unitaries $U_G$ via $\FAF_k(U_G\rho U^\dagger_G)=\FAF_k(\rho)$ and  $\FAF_k\leq n$.

We also consider another measure of fermionic non-Gaussianity, namely 
the distance to the closest FGS
\begin{equation}
    \epsilon_G(\psi)=\!
\min_{\phi_G\in\cG_n}\!
\Dtr\!\left(\psi,\phi_G\right)=\sqrt{1-\max_{\phi_G\in\cG_n}|\braket{\phi_G|\psi}|^2},
\label{eq:eGdef}
\end{equation}
where $\Dtr(\rho,\sigma)=\frac{1}{2}\Vert \rho-\sigma \Vert_1$ is the trace distance. 
The distance $\epsilon_G$ has a direct condensed-matter interpretation: it is the infidelity to the closest fermionic Gaussian state, i.e. the state-level analogue of the Hartree--Fock variational error.  Thus $\epsilon_G(\psi)$, and equivalently $\FAF_1(\psi)$ through the bounds below, quantifies the ``non-Hartree--Fockness'' of $\ket{\psi}$: the genuinely non-Gaussian many-body correlations that cannot be reproduced by any free-fermionic state. The two quantities are related as follows:
\begin{lemma}[Distance-to-FAF bound]
\label{lem:distance-faf}
For every pure state $\ket{\psi}$
\begin{equation}
\label{eq:faf-fidelity-bound-both}
\frac{1}{4n}\FAF_1(\psi)\le\epsilon_G^2(\psi)\le\frac{1}{2}\,\FAF_1(\psi).
\end{equation}
\end{lemma}
The proof is found in \SM{}~\ref{app:faf-fidelity}.
For the upper bound on $\epsilon_G^2$, we put the covariance matrix in canonical form and compare $\ket{\psi}$ with the Gaussian vacuum $\ket{\Omega}$ of the corresponding modes.  A union bound over occupations
gives
$
    1-|\langle\Omega|\psi\rangle|^2
    \le
    \frac12\sum_j(1-\nu_j^2)
    =
    \FAF_1/2.
$ 
Conversely, expanding $\ket{\psi}$ around a closest Gaussian state, chosen as the vacuum by a Gaussian change of basis, gives total occupation at most $n\epsilon_G^2(\psi)$.  The associated parity expectations are covariance matrix entries, yielding $\FAF_1\le4n\epsilon_G^2(\psi)$.

Both dependences in Eq.~\eqref{eq:faf-fidelity-bound-both} are essentially optimal: the delocalized cat state $\sqrt{1-\epsilon^2}\ket{0}^{\otimes n}+\epsilon\ket{1}^{\otimes n}$ (even $n\ge4$) saturates the $1/n$ factor, with $\FAF_1/(n\epsilon_G^2)\to4$ as $\epsilon\to0$, whereas a localized defect $\tfrac{1}{\sqrt2}(\ket{0000}+\ket{1111})\otimes\ket{0}^{\otimes(n-4)}$ has $\FAF_1=4$ and $\epsilon_G^2=1/2$ independently of $n$, ruling out any uniform bound $\epsilon_G^2\le C\,\FAF_1/n$ (see \SM{}~\ref{app:faf-fidelity}).

\prlsection{Two-copy Gaussianity testing}
While computing $\epsilon_G$ requires solving an optimization problem~\eqref{eq:eGdef}, $\FAF_k$ is especially attractive as it can be measured directly via the covariance matrix $\Gamma_\rho$~\cite{faf2025}, which can be achieved by direct measurement of the Majorana operators rewritten as Pauli operators,  via classical matchgate (MG) shadows~\cite{Wan23,zhao2021fermionic}, joint measurements of correlators~\cite{majsak2025simple}, or fermionic tomography~\cite{bittel2025optimal,mele2025efficient}, with sample and copy complexity given in Tab.~\ref{tab:testcomplexity}, see \SM{}~\ref{app:tests} for details.

We now propose a quantum protocol based on Bell measurements which enables more efficient $\FAF_1$ estimation and fermionic Gaussianity testing.
We define
$
G_a=\gamma_a\otimes\gamma_a,
$
for $a=1,\ldots,2n$.
The $G_a$ commute and are Hermitian involutions, allowing us to write
\begin{equation}
\label{eq:bellobs}
\widehat F_1
=
\frac{1}{2}
\Big(\sum_{a=1}^{2n}G_a\Big)^2.
\end{equation}

\begin{theorem}[FAF estimator]
\label{thm:bell-estimator}
For every state $\rho$,
\begin{equation}
\Tr\!\left[\widehat F_1\,\rho^{\otimes2}\right]=\FAF_1(\rho)
\end{equation}
where spectrum of $\widehat F_1$ is
\begin{equation}
\label{eq:spectrum}
\operatorname{spec}(\widehat F_1)
=
\left\{2q^2:q=0,1,\ldots,n
\right\},
\end{equation}
with largest eigenvalue $\lambda_{\max}=2n^2$.
\end{theorem}

Operationally, one applies a Bell transformation between the two copies, measures the commuting eigenvalues $g_a=\pm1$ of $G_a$, outputs $X=\frac{1}{2}\left(\sum_{a=1}^{2n} g_a\right)^2$, and averages over measurements with $\E[X]=\FAF_1(\rho)$:

% PRL: Fig.~2 (Bell-measurement circuit) removed to shorten the Letter.
% \begin{figure}[t!]
% \centering
% 	\includegraphics[width=0.24\textwidth]{Bell.pdf}
% \caption{Bell-measurement for the two-copy $\FAF_1$ estimator, shown for $n=2$ modes per copy of $\rho$.
% The measured bits are post-processed according to \Cref{prop:bell-map} to obtain the single-shot eigenvalue $\lambda_q=2q^2$.
% }
% \label{fig:bell-circuit-paper}
% \end{figure}

\begin{proposition}[Bell measurement of FAF]
\label{prop:bell-map}
Let $(u_j,v_j)\in\{0,1\}^2$ be the two measured bits from the Bell measurement on mode pair $j$, with $u_j$ the bit on the first copy and $v_j$ the bit on the second copy. Define
\begin{equation}
\begin{split}
   g_{2j-1}&= (-1)^{u_j+\sum_{k<j} v_k},\\ g_{2j} &= -(-1)^{u_j+v_j+\sum_{k<j} v_k}, 
\end{split}
\end{equation}
and $q=\frac12\sum_{a=1}^{2n}g_a$.
Then, the measured Bell sample lies in the $\widehat F_1$ eigenspace with eigenvalue
\begin{equation}
\lambda(q)=\frac{1}{2}\Bigl(\sum_{a=1}^{2n}g_a\Bigr)^2=2q^2
\end{equation}
and the FAF estimator is obtained directly from the Bell bit string. 
\end{proposition}

Since one Bell shot gives \(X\in[0,2n^2]\) with \(\mathbb E[X]=\mathrm{FAF}_1(\rho)\), 
%\old{we have $X^2\le 2n^2X$.  Together with $\FAF_1(\rho)\le n$, this gives $\operatorname{Var}(X)\le 2n^3$.  Hence averaging $N$ independent Bell outcomes estimates $\FAF_1$ with root-mean-square error at most $(2n^3/N)^{1/2}$.}  
the crude bound is $\operatorname{Var}(X)\le 2n^3$. In \SM{}~\ref{app:fourth-moment} we prove the linear fourth-moment bound
$\E[X^2]=\Tr[\widehat F_1^{\,2}\rho^{\otimes2}]\le 46\,n\,\FAF_1(\rho)$,
valid for \emph{every} (pure or mixed) state $\rho$. Together with $\FAF_1(\rho)\le n$, this gives $\operatorname{Var}(X)\le 46\,n^2$, and averaging $N$ independent Bell outcomes estimates $\FAF_1$ with root-mean-square error at most $(46\,n^2/N)^{1/2}$.
Thus additive accuracy $\eta$ with failure probability $\delta$ is achieved with
\begin{equation}
    N_{\rm shots}^{\rm Bell}
    =
    O\!\left(
        \frac{n^2}{\eta^{2}}\log\frac1\delta
    \right).
\label{eq:nshots}
\end{equation}
two-copy Bell shots, using standard median-of-means amplification.

The Bell estimator also gives a particularly easy method for fermionic Gaussianity testing~\cite{buhrman2008quantum}. Rather than reconstructing the state, the goal of property testing is to decide whether a pure state belongs to the Gaussian manifold or is separated from it by a prescribed trace distance~\cite{bittel2025optimal,mele2025efficient}.
Formally, an $(\epsilon,\delta)$-tester accepts every pure Gaussian state with probability at least $1-\delta$, and rejects every pure state whose trace distance from the Gaussian manifold is at least $\epsilon$ with probability at least $1-\delta$.
Our test is one-sided: a pure Gaussian state is accepted with probability one, because the Bell observable has deterministic zero outcome. Non-Gaussian states are rejected whenever at least one Bell shot lands outside the zero eigenspace:
\begin{theorem}[Bell Gaussianity test]
\label{thm:bell-tester}
Suppose $\ket{\psi}$ is either a pure Gaussian or satisfies $\epsilon_G(\psi)\ge\epsilon$. Then,
\begin{equation}
\label{eq:bell-test-samples}
N_{\rm shots}^{\mathrm{Bell}}=O\left(
\frac{n }{\epsilon^2}\log\frac1\delta\right)
\end{equation}
two-copy Bell measurements suffice to reject every $\epsilon$-far state with probability at least $1-\delta$, while accepting pure Gaussian states with probability $1$.
\end{theorem}
Indeed, if $\epsilon_G(\psi)\ge\epsilon$, then $\FAF_1\ge 2\epsilon^2$. 
%Since %$X\le 2n^2$,
%\[
%\Pr[X\neq0]\ge \frac{\E[X]}{2n^2}\ge \frac{\epsilon^2}{n^2}.
%\]
Since $X\geq 0$, using the fourth-moment bound of \SM{}~\ref{app:fourth-moment}, we get
\begin{equation*}
\Pr[X\neq0]\ \ge\ \frac{\E[X]^2}{\E[X^2]}\ \ge\ \frac{\E[X]}{46\,n}\ \ge\ \frac{\epsilon^2}{23\,n},
\end{equation*}
so $N=\lceil 23\,n\,\epsilon^{-2}\log(1/\delta)\rceil$ rounds suffice.

\prlsection{Single-copy Gaussianity testing}
Bell measurements require coherent control over two copies of $\rho$. Assuming we have only access to single-copy measurements, we propose another protocol with improved performance compared to previous schemes~\cite{majsak2025simple, bittel2025optimal, mele2025efficient, lyu2024convolution,coffman2025measuring}.
Our single-copy protocol builds on the fact that $\mathrm{FAF}_1(\rho)=n-\sum_{a<b}\langle B_{ab}\rangle_\rho^2$, with $B_{ab}=-i\gamma_a\gamma_b$ can be  estimated purely by squared bilinear expectation values.  The bilinears $B_{ab}$ can be partitioned into $2n-1$ measurement settings, each containing $n$ mutually commuting observables.  
In End Matter we show that measuring each setting on independent single copies and averaging products of outcomes from distinct shots enables additive estimation of $\mathrm{FAF}_1$ to precision $\eta$ with number of single-copy shots given by
\begin{equation}
    N_{\rm shots}^{\rm single}
    =
    O\!\left(
        \frac{n^3}{\eta^{2}}\log\frac1\delta
    \right),
\label{eq:nshotsSINGLE}
\end{equation}
%exhibiting the same scaling as 
a factor $n$ larger than the two-copy case~\eqref{eq:nshots}.

For Gaussianity testing, we find that single-copy testing requires more samples than Bell testing:
\begin{theorem}[Single-copy Gaussianity test]
\label{thm:single-copy-tester}
Suppose $\ket{\psi}$ is either a pure Gaussian state or satisfies
$\epsilon_G(\psi)\ge\epsilon$.  There is a single-copy measurement protocol
using
\begin{equation}
\label{eq:single-copy-test-samples}
    N_{\rm shots}^{\rm single}
    =
    O\!\left(
        \frac{n^3}{\epsilon^{4}}\log\frac1\delta
    \right)
\end{equation}
single-copy shots to reject every $\epsilon$-far state with probability at
least $1-\delta$, while accepting every pure Gaussian state with probability
at least $1-\delta$.
\end{theorem}
The theorem follows by estimating $\mathrm{FAF}_1$ to additive accuracy $\eta=\Theta(\epsilon^2)$ with our single-copy protocol and using Lemma~\ref{lem:distance-faf}; see End Matter for details.
Our protocol provides a polynomial in $n$ improvement compared to previous single-copy schemes~\cite{Wan23,zhao2021fermionic,majsak2025simple, bittel2025optimal,mele2025efficient} by measuring $n$ commuting bilinears in each setting. Importantly, the Bell test remains parametrically stronger because it is one-sided: pure Gaussian states give the zero Bell outcome deterministically, thereby avoiding the additive-estimation $\epsilon^{-4}$ penalty.

\prlsection{Purity-corrected FAF witness}
While $\FAF_1$ is faithful for pure states, it is not by itself a mixed-state non-Gaussianity witness: mixed Gaussian states generally have $\FAF_1>0$.  
Here, we introduce the purity-corrected FAF witness, analogous to mixed-state witnesses for magic and entanglement~\cite{tarabunga2025quantifying,haug2026efficient}. Whenever this witness is positive, the state is certified to be non-Gaussian.

\begin{theorem}[FAF--purity witness]
\label{thm:purity-witness}
If an $n$-qubit state $\rho$ satisfies $W_{\FAF}(\rho)>0$ with
\begin{equation}
\label{eq:wfaf-def}
W_{\FAF}(\rho)=\FAF_1(\rho)-2n\left(1-\Tr(\rho^2)^{1/n}\right)
\end{equation}
then $\rho$ is not (mixed) FGS~\eqref{eq:rhoG}.
\end{theorem}
For a Gaussian state with covariance singular values $\nu_j$ we have $\Tr(\rho^2)=\prod_{j=1}^n\frac{1+\nu_j^2}{2}$
and the bound follows immediately from the inequality of arithmetic and geometric means. 
The witness is bounded by $-n\le W_{\FAF}(\rho)\le n$ which follows from $0\le\FAF_1\le n$ and $2^{-n}\le \Tr(\rho^2)\le1$, see End Matter.  
Notably, $W_{\FAF}$ can be measured from the same two-copy data used for $\FAF_1$, because the qubit-wise Bell measurement diagonalizes both $\widehat F_1$ and the global swap operator $S$, defined by $S\ket{\alpha}\ket{\beta}=\ket{\beta}\ket{\alpha}$~\cite{bendersky2009general,garcia2013swap}.  From a Bell outcome in Proposition~\ref{prop:bell-map}, one computes the FAF sample $X$, with $\mathbb E[X]=\FAF_1(\rho)$, and the swap eigenvalue $X_{\rm pur}=(-1)^{\sum_j u_jv_j}$, with $\mathbb E[X_{\rm pur}]=\Tr(S\rho^{\otimes2})=\Tr(\rho^2)$.  Thus $W_{\FAF}$ is obtained by estimating these two means from the same measurement record and applying the nonlinear post-processing in Eq.~\eqref{eq:wfaf-def}.

Conveniently, our witness is highly robust to noise. Let us consider global depolarizing noise $\rho_p=(1-p)\ket{\psi}\!\bra{\psi}+p\frac{\id}{2^n}, 0\le p\le1$. Then, we find that if $\ket{\psi}$ is non-Gaussian, then $W_{\FAF}(\rho_p)>0$ for every $p<1$, correctly identifying $\rho_p$ as a non-Gaussian state (see \SM{}~\ref{app:depol-witness}).
We implement this witness on the IQM Garnet quantum processor for a matchgate circuit dressed with a tunable non-Gaussian gate, and find that hardware noise can both enhance and suppress non-Gaussianity in close agreement with a depolarizing-noise model (see End Matter and Fig.~\ref{fig:experiment}).

\prlsection{Pseudorandom non-Gaussianity}
The reach of non-Gaussianity probes extends beyond quantum information and many-body physics. Here we discuss an application to quantum cryptography, and in particular to the theory of pseudoresources: quantum states that contain only a small amount of a given resource yet appear, to any efficient observer, indistinguishable from states with an extensive amount. Typical Haar-random states have large quantum-resource content~\cite{Chitambar2019} and require exponentially many gates to prepare~\cite{knill1995approx}. Pseudorandom states provide a striking contrast: they can be generated efficiently, yet no polynomial-time quantum algorithm with polynomially many copies can distinguish them from Haar-random states with non-negligible advantage~\cite{ji2018pseudorandom,LaRacuente26,schuster2025random}. This computational indistinguishability has a resource-theoretic analogue: in pseudo-entanglement, pseudo-magic, and pseudo-coherence, states with only $g(n)=\omega(\log n)$ resources can mimic states with extensive $f(n)=\Theta(n)$ resource content~\cite{aaronson2022quantum,gu2023pseudomagic,haug2025pseudorandom}; see Tab.~\ref{tab:pseudo-resource-gaps}.

A popular construction is given by subset phase states~\cite{aaronson2022quantum},
\begin{equation}
\label{eq:subset-phase-main}
    \ket{\psi_{r,S}}
    =
    \frac{1}{\sqrt{2^q}}
    \sum_{x\in S}
    (-1)^{r(x)}
    \ket{x},
    \qquad
    |S|=2^q ,
\end{equation}
where $S\subseteq\{0,1\}^n$ and $r:S\to\{0,1\}$ is a phase function.  For $q=\omega(\log n)$ and quantum–secure pseudorandom functions $r$~\cite{zhandry2021construct}, these states are pseudorandom under standard cryptographic assumptions, while having only $\omega(\log n)$ entanglement, magic, and coherence.

\begin{table}[b]
\centering
\setlength{\tabcolsep}{3.5pt}
\renewcommand{\arraystretch}{0.95}
\begin{tabular}{lcc}
\hline
Pseudo-resource & $g(n)$ & $f(n)$ \\
\hline
Pseudoentanglement~\cite{aaronson2022quantum}
&
$\omega(\log n)$
&
$\Theta(n)$
\\
Pseudomagic~\cite{gu2023pseudomagic}
&
$\omega(\log n)$
&
$\Theta(n)$
\\
Pseudocoherence~\cite{haug2025pseudorandom}
&
$\omega(\log n)$
&
$\Theta(n)$
\\
Pseudo non-Gaussianity $\FAF_1$
&
$n-\operatorname{negl}(n)$
&
$n-\Theta(n^22^{-n})$
\\
\hline
\end{tabular}
\caption{Comparison of pseudo-resource gaps.  Here $g(n)$ denotes the resource
content of a pseudorandom ensemble that is computationally indistinguishable
from an ensemble with resource content $f(n)$.  For $\FAF_1$, both scalings are
nearly maximal, so any gap is negligible rather than extensive.  The notation
$\operatorname{negl}(n)$ denotes a function that decays faster than any inverse
polynomial.}
\label{tab:pseudo-resource-gaps}
\end{table}

We ask whether an analogous pseudo-resource shortcut exists for fermionic non-Gaussianity.  For $\FAF_1$, the answer is negative.  Subset phase states already have nearly maximal FAF for $q=\omega(\log n)$:
\begin{equation}
    \mathbb E[\FAF_1(\psi_{r,S})]
    \ge
    n-\frac{n(2n-1)}{2^q}
    =
    n-\operatorname{negl}(n),
\end{equation}
where $\text{negl}(n)$ denotes a negligible function which decays faster than any inverse polynomial in $n$, as we show in \SM{}~\ref{app:pseudo-nongaussianity}.  Haar-random states similarly satisfy $\FAF_1(\psi_{\rm Haar})=n-\operatorname{negl}(n)$~\cite{Turkeshi25pauli}.  More generally, any inverse-polynomial gap in $\FAF_1$ would be visible to our efficient $\FAF_1$ estimators, and hence would distinguish the ensemble from Haar random states.  Thus pseudorandomness forces Haar-like, nearly maximal $\FAF_1$; there is no pseudo-non-Gaussianity gap for this measure.

This result has an immediate consequence in near-term fault-tolerant computing, and in particular, in gate counting complexity. For circuits built from matchgates and local fermionic non-Gaussian gates~\cite{ReardonSmith24improved, Dias24classical, mele2025efficient, Paviglianiti26emergence}, preparing pseudorandom states requires at least $\Omega(n)$ non-Gaussian gates; see \SM{}~\ref{app:pseudo-nongaussianity}.  Hence, although pseudorandom states can be prepared in very small depth~\cite{schuster2025random}, their fermionic non-Gaussian gate cost must grow linearly with system size.  This mirrors the Clifford+$T$ setting, where pseudorandom-state preparation requires $\Omega(n)$ magic gates~\cite{grewal2024improved}.  The contrast with pseudo-entanglement, pseudo-magic, and pseudo-coherence stems from estimability: FAF is a low-degree correlation functional that remains efficiently measurable even at extensive values, whereas the corresponding entropic resource measures are not efficiently estimable in the same regime. 
We note that the exact value of the pseudoresource gap $f(n)$ vs $g(n)$ depends on the chosen resource measure~\cite{gu2023pseudomagic}; as such for other non-Gaussianity resource measures beyond FAF the pseudo non-Gaussianity behavior could be different.

\prlsection{Discussion}
We introduced practical witnesses and tests of fermionic magic resources built around two complementary protocols: a two-copy Bell measurement and a single-copy scheme based on commuting matchings of Majorana bilinears. Both require only Clifford operations and are thus naturally suited to early fault-tolerant devices, unlike fermionic convolution~\cite{lyu2024convolution,coffman2025measuring} or random purification~\cite{walter2025random}, which rely on costly non-Clifford resources~\cite{bravyi2005universal}. Matchgate-shadow~\cite{zhao2021fermionic,Wan23} and joint-measurement~\cite{majsak2025simple} schemes admit Clifford-compatible formulations, but at higher sampling complexity.

In parallel, the purity-corrected witness $W_{\rm FAF}$ certifies non-Gaussianity in mixed states, showing that noisy hardware can itself generate non-Gaussianity; under dephasing it can even persist in deep matchgate circuits (cf.\ \SM{}~\ref{app:numerics}). Compared with witnessing through violations of Wick's theorem~\cite{pachos2022quantifying,coffman2025measuring}, the FAF witness is more robust: we exhibit non-Gaussian states in which only one of exponentially many Wick terms is violated, yet the FAF witness always fires (see \SM{}~\ref{app:wick-vs-faf}).

Our work solves the question whether there is a gap in sampling complexity between testing and learning fermionic Gaussian states~\cite{walter2025random}. Indeed, any learning algorithm requires $\Omega(n^2)$ samples~\cite{walter2025random}, in contrast  testing is fundamentally easier with only $O(n)$ via our two-copy Bell measurement protocol.

While our algorithm has the best proven testing bound so far (Tab.~\ref{tab:testcomplexity} and \SM{}~\ref{app:tests}), its optimality remains an open problem. Curiously, there is a striking asymmetry with other resource theories:  Bell-measurement protocols allow efficient testing of entanglement~\cite{tarabunga2025quantifying,bendersky2009general,garcia2013swap}, magic~\cite{gross2021schur,haug2022scalable,haug2023efficient,haug2026efficient}, coherence~\cite{baumgratz2014quantifying,haug2025pseudorandom,streltsov2017colloquium} (see also \SM{}~\ref{app:coherence}) and bosonic Gaussianity~\cite{girardi2025gaussian} with only $O(1)$ samples~\cite{montanaro2013survey,gross2021schur,haug2025pseudorandom,girardi2025gaussian}. 
This asymmetry resurfaces at the level of cryptographic pseudoresources. For magic, entanglement, and coherence, pseudorandom states get away with only $\omega(\log n)$ resources, whereas pseudorandom fermionic states require $\Omega(n)$ non-Gaussian gates~\cite{ji2018pseudorandom,schuster2025random}: $\FAF_1$ leaves no room for a pseudo-gap. The mechanism is the efficient estimability of $\FAF_1$ via our Bell protocol: a low-degree covariance functional resolvable at extensive values would distinguish any sub-maximal ensemble from Haar-random states, whereas entropic measures evade efficient estimation there. Whether the testing and pseudoresource gaps reflect a single structural feature of free-fermion resource theory, or finer measures could close one or both, remains an appealing open question.

Beyond these questions, our protocols open experimentally accessible windows on beyond-free-fermion correlations, with natural targets in condensed matter~\cite{Bellomia26,Zavatti26corre,zavatti2026quantummagicstronglycorrelated,aditya2025growthspreadingquantumresources,aditya2025mpembaeffectsquantumcomplexity,ares2026nongaussianityrandomquantumstates}, many-body dynamics~\cite{Abanin19,Fisher23circuits,Defenu24,Sierant25mbl} where Ref.~\cite{Falcao26faf} took a first step, nuclear and high-energy physics~\cite{Robin2026directions,Santra25abelian}, and past many-body experiments~\cite{islam2015measuring,huang2022quantum,bluvstein2024logical}. Technically, the main targets are pushing the single-copy scaling below $O(n^3/\epsilon^4)$, echoing single-copy stabilizer testing~\cite{hinsche2024singlecopy}, and constructing explicit schemes for the higher-order witnesses $\FAF_k$ ($k>1$), which are intrinsically $2k$-copy observables (see \SM{}~\ref{app:higher-faf}).

\emph{Note added.} 
After completing this manuscript, we became aware of independent works addressing single-copy~\cite{poetri2026construct} and two-copy~\cite{tarabunga2026fermionicnongaussianitybellsampling} tests of fermionic non-Gaussianity, and FAF estimation protocol~\cite{leone2026an} that matches the scaling of Eq.~\eqref{eq:nshots}.

\begin{acknowledgments}

\textbf{Acknowledgments.}
We thank Stefano Cusumano, Guglielmo Lami, Lorenzo Leone, Zhenhuan Liu, Salvatore Francesco Emanuele Oliviero, Ingo Roth, Paolo Stornati, Bujiao Wu, Poetri Sonya Tarabunga, and Emanuele Tirrito for enlightening discussions, and M. Walter and F. Witteveen for comments on the early version of this manuscript.
The Python code for our work is available on GitHub~\cite{haug2025fermiontests}.
X.T. acknowledges support from DFG Emmy Noether Programme proposal ``\textit{Digital Quantum Matter Ouf-of-Equilibrium}'' No. 560726973, DFG under Germany's Excellence Strategy – Cluster of Excellence Matter and Light for Quantum Computing (ML4Q) EXC 2004/2 – 390534769, and DFG Collaborative Research Center (CRC) 183 Project No. 277101999 - project B01. 
P.S. acknowledges fellowship within the “Generación D” initiative, Red.es, Ministerio para la Transformación Digital y de la Función Pública, for talent attraction (C005/24-ED CV1), funded by the European Union NextGenerationEU funds, through PRTR.

\end{acknowledgments}

%\bibliography{fermionic_magic_resources}
\input{output.bbl}

%=========================== END MATTER ===========================
\onecolumngrid
\begin{center}
\vspace{0.35cm}
\textbf{\large End Matter}
\end{center}
\twocolumngrid
\setcounter{secnumdepth}{2}
\setcounter{section}{0}
\renewcommand{\thesection}{\Alph{section}}
\setcounter{figure}{1}

\prlsection{Experimental implementation on a noisy quantum processor}
\label{app:experiment}
We experimentally measure our non-Gaussianity witness $W_\text{FAF}$ on the IQM Garnet quantum computer in Fig.~\ref{fig:experiment}. We study a matchgate circuit combined with a non-Gaussian gate parameterized with angle $\theta$ which controls the amount of injected non-Gaussianity, with $\theta=0$ corresponding to a fully matchgate circuit, while $\theta=\pi/2$ will induce the maximal $W_{\text{FAF}}=4$ (see Fig.~\ref{fig:experiment}a). Our experiment in Fig.~\ref{fig:experiment}b matches closely a noisy simulation subject to local depolarizing noise. Notably, the noise of the quantum computer can both increase and reduce non-Gaussianity: For small $\theta$, we observe an increase in $W_{\text{FAF}}$ compared to the noise-free simulation, while for large $\theta$ we observe a lower $W_{\text{FAF}}$.
In fact, in simulations of Fig.~\ref{fig:experiment}c, we find this behavior persists for $\theta\lesssim \pi/5$, where for small $p$, increasing noise yields an increase in  $W_{\text{FAF}}$, while for larger $\theta$ and $p$ we observe a monotonic decay towards non-positive $W_{\text{FAF}}$. This noise-induced non-Gaussianity can also be observed across different noise-models, such as dephasing or amplitude damping for general matchgate circuit constructions, and persist even in the steady-state of very deep circuits (see \SM{}~\ref{app:numerics}).

\begin{figure*}[t!]
	\centering
    \includegraphics[width=0.95\textwidth]{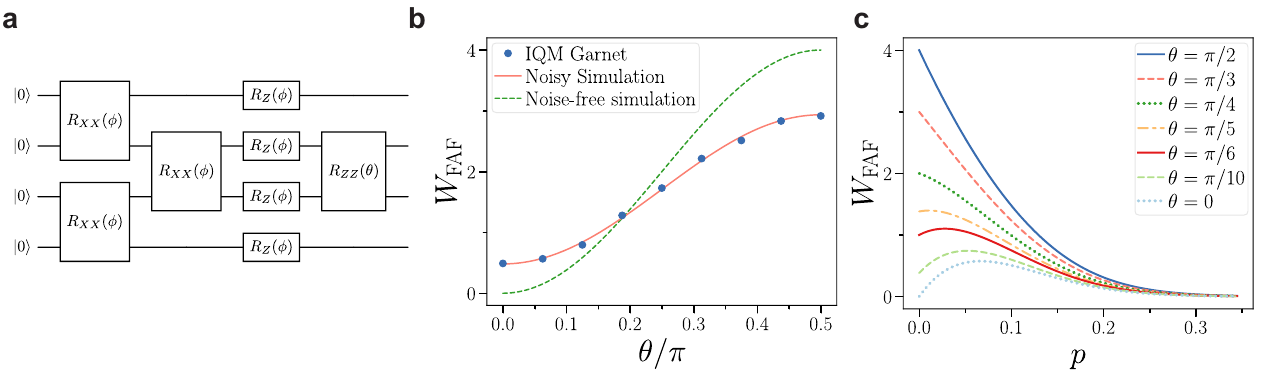}
\vspace{-0.3cm}
\caption{Experimental measurement of fermionic non-Gaussianity witness $W_\text{FAF}$ on IQM Garnet quantum computer. \idg{a} State preparation circuit with $n=4$ qubits of parameterized matchgates $R_{xx}=\exp(-i\phi/2 X\otimes X)$  and $R_{z}=\exp(-i\phi/2 Z)$ with $\phi=\pi/2$ and non-Gaussian gate $R_{zz}=\exp(-i\theta/2 Z\otimes Z)$. \idg{b} $W_\text{FAF}$ against $\theta$ for implementation with IQM Garnet quantum computer (blue), noisy simulation (orange) and noise-free simulation (green). Noisy simulation assumes single-qubit depolarizing noise $\Lambda_p=(1-p)\rho+p\frac{\id}{2}$ applied after every two-qubit gate with $p=0.036$. Each experimental datapoint is extracted from $N_\text{shots}=20000$ Bell measurements.
\idg{c} Simulation of $W_\text{FAF}$ against $p$ for different $\theta$. We find that for small $\theta \lesssim \pi/5$, noise $p$ can yield an increase in $W_\text{FAF}$, while for larger $\theta$ the witness decreases with $p$.
}
	\label{fig:experiment}
\end{figure*}

%\prlsection{FAF estimation and Gaussianity testing}
%\label{app:additional-proofs}
\prlsection{Bell estimator.}
Let $G_a=\gamma_a\otimes\gamma_a$ and $K=\sum_{a=1}^{2n}G_a$.  The operators $G_a$ commute pairwise: for $a\neq b$, the two minus signs from exchanging Majoranas in the two copies cancel.  Moreover $G_a^2=\id$, and hence
\( K^2=2n\,\id+2\sum_{a<b}\gamma_a\gamma_b\otimes\gamma_a\gamma_b .
\)
Denoting $B_{ab}=-i\gamma_a\gamma_b$ and $\Gamma_{ab}=\Tr(\rho B_{ab})$, we have $\Tr(\rho\,\gamma_a\gamma_b)=i\Gamma_{ab}$ for $a<b$.  Therefore
\[
    \Tr(K^2\rho^{\otimes2})
    =
    2n
    -
    2\sum_{a<b}\Gamma_{ab}^2
    =
    2n-\|\Gamma_\rho\|_F^2,
\]
where \(\|\Gamma_\rho\|_F^2\equiv
\sum_{a,b=1}^{2n}(\Gamma_\rho)_{ab}^2
=2\sum_{a<b}(\Gamma_\rho)_{ab}^2\).
Thus, for $\widehat F_1:=K^2/2$,
\[
    \Tr(\widehat F_1\rho^{\otimes2})
    =
    n-\frac12\|\Gamma_\rho\|_F^2
    =
    \FAF_1(\rho).
\]
Hence the Bell observable $\widehat F_1$ gives an unbiased estimator of the $\FAF_1$.

The operators $G_a$ are Hermitian and commute pairwise, and hence admit a common eigenbasis.  In this basis, $G_a\ket{g}=g_a\ket{g}$ with $g_a=\pm1$, where the eigenvalues are $\pm1$ because $G_a^2=\id$.  Therefore $K\ket{g}=(\sum_a g_a)\ket{g}$.  The sum of $2n$ variables $g_a$ is an even integer $2s$, with $s=-n,\ldots,n$, so the possible outcomes of $\widehat F_1$ measurement are $2q^2$, with $q=0,\ldots,n$, resulting in Theorem~\ref{thm:bell-estimator}.

\prlsection{Bell-bit readout.}
The eigenvalues $g_a$ are obtained from qubit-wise Bell measurements.  For mode $j$, after a CNOT from copy $1$ to copy $2$ followed by a Hadamard on copy $1$, let $(u_j,v_j)$ be the two measured bits.  With this convention, the measured eigenvalues of the commuting two-copy Pauli operators
are  $X_j^{(1)}X_j^{(2)} \mapsto (-1)^{u_j}$ and $Z_j^{(1)}Z_j^{(2)} \mapsto (-1)^{v_j}$, while $Y_j^{(1)}Y_j^{(2)} \mapsto -(-1)^{u_j+v_j}$.  Combining this with the Jordan--Wigner strings form~\cite{Jordan1928} gives
\[
    g_{2j-1}
    =
    (-1)^{u_j+\sum_{k<j}v_k},
    \qquad
    g_{2j}
    =
    -(-1)^{u_j+v_j+\sum_{k<j}v_k}.
\]
The measured Bell string therefore determines the Bell-estimator outcome as $\lambda=\frac12(\sum_a g_a)^2$, justifying Proposition~\ref{prop:bell-map}.

\prlsection{Definition of testing.} We give a formal definition of testing~\cite{buhrman2008quantum}:
\begin{definition}[Fermionic Gaussianity testing]
An \((\epsilon,\delta)\)-tester for pure-state fermionic Gaussianity is a
measurement protocol which, given copies of an unknown pure state
\(|\psi\rangle\), satisfies the following two properties:
\begin{align*}
&(i)\quad |\psi\rangle\in\mathcal G_n
&&\Longrightarrow\quad
\Pr[\mathrm{accept}]\ge 1-\delta ,\\
&(ii)\quad \epsilon_G(\ket{\psi})\ge \epsilon
&&\Longrightarrow\quad
\Pr[\mathrm{reject}]\ge 1-\delta .
\end{align*}
\end{definition}

\prlsection{FAF--purity witness.}
For a mixed FGS with covariance singular values $\nu_j\in[0,1]$, the purity and antiflatness are $\Tr(\rho^2)=\prod_j(1+\nu_j^2)/2$ and $\FAF_1=n-\sum_j\nu_j^2$.  The arithmetic--geometric mean inequality applied to the factors
$(1+\nu_j^2)/2$ gives
\[
    \Tr(\rho^2)^{1/n}
    \le
    \frac1n\sum_{j=1}^n\frac{1+\nu_j^2}{2}
    =
    1-\frac{\FAF_1}{2n}.
\]
Equivalently, every mixed Gaussian state satisfies
\(    \FAF_1(\rho)
    \le
    2n\left(1-\Tr(\rho^2)^{1/n}\right).
\)
Thus the quantity $W_{\FAF}(\rho):=\FAF_1(\rho)-2n(1-\Tr(\rho^2)^{1/n})$ is non-positive on all mixed Gaussian states, and any positive value certifies fermionic non-Gaussianity, leading to Theorem~\ref{thm:purity-witness}.

% \vspace{0.2cm}

\prlsection{Single-copy estimator for $\FAF_1$}
\label{app:single-copy-methods}
We describe the single-copy estimator of the antiflatness
\(
    \FAF_1(\rho)
    =
    n-\sum_{1\le a<b\le 2n}\langle B_{ab}\rangle_\rho^2,
\)
denoting $B_{ab}\equiv-i\gamma_a\gamma_b$ .
The key observation is that two distinct bilinears $B_{ab}$ and $B_{cd}$
commute whenever the Majorana pairs $\{a,b\}$ and $\{c,d\}$ are disjoint.
Thus the bilinears can be grouped into commuting measurement settings.
We use the following explicit decomposition.  Let $m=2n-1$,
$[x]_m:=1+((x-1)\bmod m)$.  For
$u\neq v$, write
$\operatorname{ord}(u,v):=(\min\{u,v\},\max\{u,v\})$.  For
$\ell=1,\ldots,m$, define
\begin{equation}
    \mathcal M_\ell
    \!=\!
    \{(\ell,2n)\}
    \cup
    \left\{
        \operatorname{ord}\big([\ell+j]_m,[\ell-j]_m\big) \!:\!
        j\!=\!1,\ldots,n-1
    \right\}.
    \label{eq:measurement-layers}
\end{equation}
Each $\mathcal M_\ell$ contains $n$ disjoint pairs, and the layers form a
partition
$
    \big\{(a,b):1\le a<b\le 2n\big\}
    =
    \bigsqcup_{\ell=1}^{2n-1}\mathcal M_\ell .
    \label{eq:layer-partition}
$
Indeed, pairs involving $2n$ appear as $(\ell,2n)$, while a pair $(a,b)$ with $1\le a<b\le m$ belongs to the unique layer satisfying $2\ell\equiv a+b\pmod m$. 
In a Jordan--Wigner encoding, each $B_{ab}$ is a Pauli string.  Hence each layer $\mathcal M_\ell$ is an abelian Pauli measurement: a Clifford circuit can map the $n$ commuting strings $\{B_e:e\in\mathcal M_\ell\}$ to single-qubit $Z$ measurements, after which the outcomes give the vector $X_r$ used in Eq.~\eqref{eq:Sl-ustat}.  Thus the single-copy protocol requires only Clifford operations and computational-basis measurements.

For each layer $\mathcal M_\ell$, we perform $M$ independent single-copy measurements of the commuting family $\{B_e:e\in\mathcal M_\ell\}$.  Let $X_r=(X_r^e)_{e\in\mathcal M_\ell}\in\{\pm1\}^n$ be the outcome vector in shot $r$, and set $\mu_e:=\langle B_e\rangle_\rho$ and $S_\ell:=\sum_{e\in\mathcal M_\ell}\mu_e^2$.  We estimate $S_\ell$ by
\begin{equation}
    \widehat S_\ell
    =
    \binom{M}{2}^{-1}
    \sum_{1\le r<s\le M}
    \sum_{e\in\mathcal M_\ell}
    X_r^eX_s^e .
    \label{eq:Sl-ustat}
\end{equation}
Since different shots are independent,
$\mathbb E[X_r^eX_s^e]=\mu_e^2$ for $r\neq s$, and therefore
$\mathbb E[\widehat S_\ell]=S_\ell$.  Hence
$
    \widehat{\FAF}_1
    :=
    n-\sum_{\ell=1}^{2n-1}\widehat S_\ell
    \label{eq:F1hat}
$
is an unbiased estimator of $\FAF_1(\rho)$.

We now bound its variance.  For a fixed layer, define $h(X,Y):=\sum_{e\in\mathcal M_\ell}X^eY^e$ for two independent outcome vectors $X,Y$.  In the variance of the pair average \eqref{eq:Sl-ustat}, the only nonzero contributions come from identical shot
pairs and from pairs sharing one shot.  Thus
\(
    \operatorname{Var}(\widehat S_\ell)
    =
    \binom{M}{2}^{-1}\big(2(M-2)\zeta_1+\zeta_2\big),
\)
where
$\zeta_1=\operatorname{Cov}(h(X,Y),h(X,Z))$ and $\zeta_2=\operatorname{Var}(h(X,Y))$, with $X,Y,Z$ independent.  Moreover, since \(Y\) is independent of \(X\) and
\(\mathbb E[Y^e]=\mu_e\), we have \(\mathbb E[h(X,Y)\mid X]=\sum_{e\in\mathcal M_\ell}\mu_e X^e\). Thus, writing \(X,Y,Z\) for independent layer outcomes,
\begin{equation}
\begin{aligned}
    \zeta_1
    &=
    \mathbb E\!\left[h(X,Y)h(X,Z)\right]
    -
    \mathbb E[h(X,Y)]\,\mathbb E[h(X,Z)] \\
    &=
    \mathbb E\!\left[
        \mathbb E[h(X,Y)\mid X]\,
        \mathbb E[h(X,Z)\mid X]
    \right]
    -
    S_\ell^2  \\
    &=
    \mathbb E\!\left[
        \left(\sum_{e\in\mathcal M_\ell}\mu_e X^e\right)^2
    \right]
    -
    S_\ell^2
    \le
    \mathbb E\!\left[
        \left(\sum_{e\in\mathcal M_\ell}\mu_e X^e\right)^2
    \right]
.
\label{eq:inter}
\end{aligned}
\end{equation}
Here the second line uses the conditional independence of \(h(X,Y)\) and \(h(X,Z)\) given \(X\), and using the Cauchy--Schwarz inequality together with \((X^e)^2=1\) to bound the final expectation value in \eqref{eq:inter}, we find $\zeta_1 \le nS_\ell$.
Consequently,
\begin{equation}
    \operatorname{Var}(\widehat S_\ell)
    \lesssim
    \frac{nS_\ell}{M}
    +
    \frac{n^2}{M^2}.
    \label{eq:Sl-var-bound}
\end{equation}

Different layers are measured on independent batches of shots.  Using
$\sum_\ell S_\ell=\sum_{a<b}\langle B_{ab}\rangle_\rho^2
=\|\Gamma_\rho\|_F^2/2\le n$, we obtain
$
    \operatorname{Var}(\widehat{\FAF}_1)
    \lesssim
    \frac{n^2}{M}
    +
    \frac{n^3}{M^2}.
    \label{eq:FAF1-var-bound}
$
Thus root-mean-square additive accuracy $\eta$ is achieved with
$M=O(n^2\eta^{-2}+n^{3/2}\eta^{-1})$ shots per layer, or
\begin{equation}
    N_{\rm shots}
    =
    O\!\left(
        n^3\eta^{-2}
        +
        n^{5/2}\eta^{-1}
    \right)
    \label{eq:Ntot-eta}
\end{equation}
single-copy measurements in total.  For $\eta\lesssim\sqrt n$, the first
term dominates and Eq.~\eqref{eq:nshotsSINGLE} follows.  A standard median-of-independent-runs amplification adds a factor $O(\log(1/\delta))$ for failure probability $\delta$.

For pure-state Gaussianity testing, Lemma~\ref{lem:distance-faf} gives
$\epsilon_G(\psi)^2\le \FAF_1(\psi)/2$.  Hence every
$\epsilon$-far pure state satisfies $\FAF_1(\psi)\ge2\epsilon^2$.  Taking
$\eta=\Theta(\epsilon^2)$ in Eq.~\eqref{eq:Ntot-eta} gives
\begin{equation}
    N_{\rm shots}
    =
    \widetilde O\!\left(
        n^3\epsilon^{-4}
        +
        n^{5/2}\epsilon^{-2}
    \right)
    =
    \widetilde O(n^3\epsilon^{-4}),
    \label{eq:Ntot-eps}
\end{equation}
where logarithmic factors in the failure probability are suppressed. This demonstrates Theorem~\ref{thm:single-copy-tester}.

\let\addcontentsline\oldaddcontentsline

%======================= SUPPLEMENTAL MATERIAL =======================
\appendix

\onecolumngrid

\newpage

%\appendix
\setcounter{secnumdepth}{2}

\renewcommand{\thesection}{\Alph{section}}
\renewcommand{\thesubsection}{\arabic{subsection}}
%\renewcommand*{\theHsection}{\thesection}

%\makeatletter
%\renewcommand{\l@section}[2]{%
%  \@dottedtocline{1}{1.5em}{2.8em}{#1}{#2}%
%}
%\makeatother

%\titleformat{\section}[hang]{\normalfont\bfseries}{\SM~\thesection:}{0.5em}{\centering}

%\setcounter{secnumdepth}{1}

\clearpage
\begin{center}

\textbf{\large Supplemental Information}
\end{center}
%\setcounter{equation}{0}
%\setcounter{figure}{0}
%\setcounter{table}{0}

%\makeatletter
%\renewcommand{\theequation}{S\arabic{equation}}
%\renewcommand{\thefigure}{S\arabic{figure}}
%\renewcommand{\bibnumfmt}[1]{[S#1]}
%\newtheorem{thmS}{Theorem S\ignorespaces}
%\newtheorem{lemmaS}{Lemma S\ignorespaces}

%\newtheorem{claimS}{Claim S\ignorespaces}

%\newtheorem{definitionS}{Definition S\ignorespaces}

In the Supplemental Material, we provide further proofs underlying our main results as well as additional findings.

\tableofcontents

\section{FAF and Gaussian fidelity}
\label{app:faf-fidelity}

We prove the two-sided fidelity bounds and the localized counterexample of the main text. 
Let us recall the bound:
\begin{equation}
\label{eq:faf-fidelity-bound-both_SM}
\frac{1}{4n}\FAF_1(\psi)\le\epsilon_G^2(\psi)\le\frac{1}{2}\,\FAF_1(\psi).
\end{equation}

Let us begin with the right-hand side of Eq.~\eqref{eq:faf-fidelity-bound-both_SM}. Since both \(\epsilon_G(\psi)\) and \(\FAF_1(\psi)\) are invariant under Gaussian unitaries, we may work in a canonical Majorana basis \(\{\widetilde\gamma_a\}\) in which
\begin{equation}
    \Gamma_\psi
    =
    \bigoplus_{j=1}^n
    \begin{pmatrix}
        0 & \nu_j\\
        -\nu_j & 0
    \end{pmatrix},
    \qquad
    0\le \nu_j\le 1 .
\end{equation}
We  define the fermionic annihilation operators
\(c_j=(\widetilde\gamma_{2j-1}+i\widetilde\gamma_{2j})/2\). 
With this convention,
\(-i\widetilde\gamma_{2j-1}\widetilde\gamma_{2j}=1-2c_j^\dagger c_j\), and hence the canonical form of \(\Gamma_\psi\) gives
\begin{equation}
    \langle c_j^\dagger c_j\rangle_\psi=\frac{1-\nu_j}{2},
\end{equation}
where \(\langle O\rangle_\psi:=\langle\psi|O|\psi\rangle\) for any operator $O$. 
Let \(\ket{\Omega}\) be the Gaussian vacuum annihilated by all \(c_j\), and set \(n_j:=c_j^\dagger c_j\). Since the \(n_j\)'s are commuting projectors, they are simultaneously diagonal in the occupation basis. In that basis, \(I-\ket{\Omega}\bra{\Omega}\) is the indicator of the event that at least one mode is occupied, while \(\sum_j n_j\) counts the number of occupied modes.
Hence
\begin{equation}
    I-\ket{\Omega}\bra{\Omega}
    \le
    \sum_{j=1}^n n_j.
\end{equation}
Equivalently, on an occupation string \(\ket{x_1,\ldots,x_n}\), the left-hand side has eigenvalue \(0\) for the vacuum string \(x_1=\cdots=x_n=0\), and eigenvalue \(1\) for every non-vacuum string, whereas the right-hand side has eigenvalue \(\sum_j x_j\).
Taking the expectation value in \(\ket{\psi}\), we obtain
\begin{equation}
    1-|\langle\Omega|\psi\rangle|^2
    =
    \langle I-\ket{\Omega}\bra{\Omega}\rangle_\psi
    \le
    \sum_{j=1}^n \langle c_j^\dagger c_j\rangle_\psi
    =
    \frac12\sum_{j=1}^n(1-\nu_j)
    \le
    \frac12\sum_{j=1}^n(1-\nu_j^2)
    =
    \frac12\FAF_1(\psi).
\end{equation}
Since $\ket{\Omega}$ is one candidate in the maximization defining $\epsilon_G(\psi)$, we have
\[
\epsilon_G^2(\psi)\le 1-|\braket{\Omega|\psi}|^2
\le \frac 12\FAF_1(\psi),
\]
which is the right-hand side of Eq.~\eqref{eq:faf-fidelity-bound-both_SM}. We note that this bound is tighter by a factor of $n$ compared to a bound derived in~\cite{Vershynina14}.

We next prove the converse bound.
Let \(\ket{\phi_G}\) be a Gaussian state realizing the minimum in the definition of \(\epsilon_G(\psi)\). By applying a Gaussian unitary, which leaves both \(\FAF_1\) and \(\epsilon_G\) unchanged, we may choose coordinates in which \(\ket{\phi_G}\) is the Fock vacuum \(\ket{\Omega}\). Thus
\begin{equation}
    \ket{\psi}
    =
    \alpha\ket{\Omega}
    +
    \epsilon_G\ket{\eta},
    \qquad
    \epsilon_G=\epsilon_G(\psi),
    \qquad
    \braket{\Omega|\eta}=0, \qquad
    \braket{\eta|\eta}=1.
\end{equation}
Let \(c_j\) be the corresponding vacuum modes, \(n_j=c_j^\dagger c_j\), and \(p_j:=\langle n_j\rangle_\psi\). Since the total number operator \(N=\sum_j n_j\) satisfies \(0\le N\le nI\) and annihilates \(\ket{\Omega}\),
we have
\begin{equation}
    \sum_{j=1}^n p_j
    =
    \langle N\rangle_\psi
    =
    \epsilon_G^2\langle\eta|N|\eta\rangle
    \le
    n\epsilon_G^2 .
\end{equation}
In this same mode basis, the covariance matrix contains the \(n\) entries
\[
    (\Gamma_\psi)_{2j-1,2j}
    =
    \langle -i\gamma_{2j-1}\gamma_{2j}\rangle_\psi
    =
    1-2p_j .
\]
Therefore,
\begin{align}
    \FAF_1(\psi)
    =
 n-\sum_{a<b}((\Gamma_\psi)_{ab})^2
    \le
    n-\sum_{j=1}^n(1-2p_j)^2 \
    =
    4\sum_{j=1}^n p_j(1-p_j)
    \le
    4\sum_{j=1}^n p_j
    \le
    4n\epsilon_G^2 .
\end{align}
This proves the converse bound.

The constant scaling on the right-hand side of Eq.~\eqref{eq:faf-fidelity-bound-both_SM}  cannot be improved below $4$ in general. For even $n\ge4$, consider the superposition of the vacuum state $\ket{0}^{\otimes n}$ and the fully occupied fermionic Gaussian state $\ket{1}^{\otimes n}$:
\begin{equation}
\ket{\psi_{n,\epsilon}}=\sqrt{1-\epsilon^2}\ket{0}^{\otimes n}+\epsilon\ket{1}^{\otimes n},
\qquad 0\le\epsilon^2\le\frac12.
\end{equation}
In this state, each canonical mode is occupied with probability
\(\epsilon^2\), i.e. \(\langle n_j\rangle_{\psi_{n,\epsilon}}=\epsilon^2\).
Moreover, for \(n\ge4\), the off-diagonal matrix elements
\(\bra{0}^{\otimes n}(-i\gamma_a\gamma_b)\ket{1}^{\otimes n}\) vanish for all
\(a<b\): a quadratic Majorana operator can change the occupation pattern in at
most two modes, whereas \(\ket{0}^{\otimes n}\) and \(\ket{1}^{\otimes n}\)
differ in all \(n\) modes. Hence the two-point covariance matrix receives no
coherent contribution from the superposition between the two components.

It follows that the only nonzero covariance entries are the diagonal canonical pair entries,
\[
    \langle -i\gamma_{2j-1}\gamma_{2j}\rangle_{\psi_{n,\epsilon}}
    =
    1-2\epsilon^2,
\]
so all covariance singular values are \(|1-2\epsilon^2|\).

Moreover, let $\ket{\phi}$ be an arbitrary even-parity pure Gaussian state. By the fermionic Bloch--Messiah decomposition~\cite{bloch1962canonical}, it can be written as
\begin{equation}
    \ket{\phi}
    =
    U_1
    \prod_{k=1}^{n/2}
    \big(
        u_k+v_k c_{2k-1}^\dagger c_{2k}^\dagger
    \big)
    \ket{0}^{\otimes n},
    \qquad
    |u_k|^2+|v_k|^2=1,
\end{equation}
where $U_1$ is a number-conserving Gaussian unitary. Since $U_1$ fixes the vacuum and maps the fully occupied Slater determinant to itself up to a determinant phase, the vacuum and fully occupied amplitudes obey
\begin{equation}
    \left|\bra{0^{\otimes n}}\phi\rangle\right|
    =
    \prod_{k=1}^{n/2}|u_k|,
    \qquad
    \left|\bra{1^{\otimes n}}\phi\rangle\right|
    =
    \prod_{k=1}^{n/2}|v_k|.
\end{equation}
For $n\ge4$, i.e. for at least two BCS pairs, iterated Cauchy--Schwarz gives
\begin{equation}
    \prod_{k=1}^{n/2}|u_k|
    +
    \prod_{k=1}^{n/2}|v_k|
    \le
    1 .
\end{equation}
Indeed, for two pairs this is the ordinary Cauchy--Schwarz inequality,
\[
    |u_1u_2|+|v_1v_2|
    \le
    \sqrt{|u_1|^2+|v_1|^2}
    \sqrt{|u_2|^2+|v_2|^2}
    =
    1,
\]
and the extension to more pairs follows by applying the same argument inductively. Therefore, writing $a=\sqrt{1-\epsilon^2}$ and $b=\epsilon$, with $a\ge b$ for $\epsilon^2\le1/2$, we have
\begin{align}
    \left|\bra{\psi_{n,\epsilon}}\phi\rangle\right|
    \le
    a\left|\bra{0^{\otimes n}}\phi\rangle\right|
    +
    b\left|\bra{1^{\otimes n}}\phi\rangle\right|  
    \le
    a\left(
        \left|\bra{0^{\otimes n}}\phi\rangle\right|
        +
        \left|\bra{1^{\otimes n}}\phi\rangle\right|
    \right)
    \le
    \sqrt{1-\epsilon^2}.
\end{align}
The bound is achieved by the Gaussian vacuum. Hence the closest Gaussian state is the vacuum and $\epsilon_G(\psi_{n,\epsilon})=\epsilon$. Consequently, $\epsilon_G(\psi_{n,\epsilon})=\epsilon$ and
\[
\qquad
\frac{\FAF_1(\psi_{n,\epsilon})}{n \epsilon_G^2(\psi_{n,\epsilon})}=4(1-\epsilon^2),
\]
which tends to $4$ as $\epsilon\downarrow0$.

We next discuss a localized non-Gaussian defect, which shows that the upper bound \(\epsilon_G^2(\psi)\lesssim \FAF_1(\psi)\) cannot, in general, be improved by any factor that scales with the system size. Let
\[
\ket{\chi_4}=\frac{\ket{0000}+\ket{1111}}{\sqrt2}
\]
on four modes and set $\ket{\psi_n}=\ket{\chi_4}\otimes\ket{0}^{\otimes(n-4)}$. In the active four-mode block, all two-point Majorana correlators vanish: the only coherence between $\ket{0000}$ and $\ket{1111}$ is a four-body coherence, and each occupation parity has expectation zero. Thus the four active covariance singular values are zero, while each spectator vacuum mode contributes one singular value equal to one, resulting in 
$\FAF_1(\psi_n)=4$.

It remains to compute the Gaussian fidelity. Consider an arbitrary even pure Gaussian state on four modes, which can be written as
\[
\ket{\phi_Z}=\frac{
\exp\!\left(\frac12\sum_{i,j=1}^4 Z_{ij}c_i^\dagger c_j^\dagger\right)\ket{0}
}{\det(1+Z^\dagger Z)^{1/4}},
\qquad Z=-Z^T .
\]
By the fermionic Bloch--Messiah decomposition, after a number-conserving Gaussian rotation the pairing matrix has two singular values $z_1,z_2$, and the state takes the form
\[
\ket{\phi}=(u_1+v_1 c_1^\dagger c_2^\dagger)(u_2+v_2 c_3^\dagger c_4^\dagger)\ket{0},
\qquad |u_i|^2+|v_i|^2=1 .
\]
The overlap with $\ket{\chi_4}$ is bounded by
\begin{align}
|\braket{\chi_4|\phi}|^2
=\frac12|u_1u_2+v_1v_2|^2 
&\le \frac12 (|u_1u_2|+|v_1v_2|)^2 \\
&\le \frac12(|u_1|^2+|v_1|^2)(|u_2|^2+|v_2|^2)=\frac12,
\end{align}
where the last step is Cauchy--Schwarz inequality. The bound is achieved by the Gaussian vacuum $\ket{0000}$, so $\epsilon_G^2(\chi_4)=1/2$.

Appending Gaussian spectator vacua does not increase the maximum Gaussian overlap with the active block. Indeed, projecting any $n$-mode Gaussian state onto the spectator vacuum produces, after normalization when nonzero, a Gaussian state on the active four modes; hence the best overlap with $\ket{\psi_n}$ is bounded by the best four-mode Gaussian overlap with $\ket{\chi_4}$. Conversely, the product Gaussian $\ket{0000}\otimes\ket{0}^{\otimes(n-4)}$ achieves overlap squared $1/2$. Therefore $\epsilon_G^2(\psi_n)=1/2$. This proves the counter-example and shows that no linear-scaling inequality of the form $\FAF_1\ge cn\epsilon_G^2$ can hold for FAF.

\section{Linear fourth-moment bound for the Bell estimator}
\label{app:fourth-moment}

In this section we prove the fourth-moment bound quoted in the main text, which improves the Bell-estimator variance from $O(n^3)$ to $O(n^2)$ and the Bell-tester sample complexity from $O(n^2/\epsilon^2)$ to $O(n/\epsilon^2)$.  Throughout, $K=\sum_{a=1}^{2n}G_a$ with $G_a=\gamma_a\otimes\gamma_a$ as in the main text, $\widehat F_1=K^2/2$, and we write
\begin{equation}
m_k:=\Tr\!\big[K^k\rho^{\otimes2}\big],\qquad k=2,4,
\end{equation}
so that $m_2=2\,\FAF_1(\rho)$ (End Matter) and a Bell sample $X$ has moments $\E[X]=m_2/2$, $\E[X^2]=m_4/4$.
 
\begin{theorem}[Fourth moment of the Bell observable]
\label{thm:fourth-moment}
For every $n$-mode state $\rho$, pure or mixed,
\begin{equation}
\label{eq:fourth-moment}
\E[X^2]\;=\;\Tr\!\big[\widehat F_1^{\,2}\,\rho^{\otimes2}\big]\;\le\;46\,n\;\Tr\!\big[\widehat F_1\,\rho^{\otimes2}\big]\;=\;46\,n\,\FAF_1(\rho).
\end{equation}
\end{theorem}
 
This has two immediate consequences.  First, $\operatorname{Var}(X)\le\E[X^2]\le46\,n\,\FAF_1(\rho)\le46\,n^2$, which yields Eq.~\eqref{eq:nshots} of the main text.  Second, since $X\ge0$, Cauchy--Schwarz on the event $\{X>0\}$ gives $\E[X]^2\le\Pr[X\neq0]\,\E[X^2]$, hence
\begin{equation}
\Pr[X\neq0]\;\ge\;\frac{\E[X]}{46\,n}\;=\;\frac{\FAF_1(\rho)}{46\,n},
\end{equation}
which, combined with $\FAF_1\ge2\epsilon_G^2$ for pure states (Lemma~\ref{lem:distance-faf}), proves Theorem~\ref{thm:bell-tester} with $N=\lceil23\,n\,\epsilon^{-2}\log(1/\delta)\rceil$. 
 
\subsection{An exact identity for commuting involutions}
 
\begin{lemma}
\label{lem:inv4}
Let $J_1,\dots,J_N$ be pairwise commuting involutions ($J_a^2=\id$) and $S=\sum_a J_a$.  With $e_4=\sum_{a<b<c<d}J_aJ_bJ_cJ_d$,
\begin{equation}
\label{eq:inv4}
S^4=(6N-8)\,S^2-3N(N-2)\,\id+24\,e_4 .
\end{equation}
\end{lemma}
 
\begin{proof}
$S^2=N\,\id+2e_2$ with $e_2=\sum_{a<b}J_aJ_b$.  Squaring $e_2$ and classifying ordered pairs of pairs by their overlap (identical pairs: $\binom N2$ terms equal to $\id$; pairs sharing one index: each product $J_bJ_c$ arises $2(N-2)$ times; disjoint pairs: each 4-subset arises $6$ times) gives $e_2^2=\binom N2\,\id+2(N-2)\,e_2+6\,e_4$.  Substituting $e_2=(S^2-N)/2$ into $S^4=(N+2e_2)^2$ yields Eq.~\eqref{eq:inv4}.
\end{proof}
 
Applying Lemma~\ref{lem:inv4} with $N=2n$ to the commuting involutions $J_a=G_a$ and taking $\Tr[\,\cdot\,\rho^{\otimes2}]$, note that $\Tr[G_aG_bG_cG_d\,\rho^{\otimes2}]=\big(\Tr[\rho\,\gamma_a\gamma_b\gamma_c\gamma_d]\big)^2$, and that the ordered product of four distinct Majoranas is Hermitian, so these squares are non-negative.  Defining $A_4:=\sum_{a<b<c<d}\big(\Tr[\rho\,\gamma_a\gamma_b\gamma_c\gamma_d]\big)^2$, we get
\begin{equation}
\label{eq:exact-m4}
m_4=(12n-8)\,m_2+24\Big(A_4-\tbinom n2\Big).
\end{equation}
Theorem~\ref{thm:fourth-moment} therefore follows once we show $A_4-\binom n2\le\tfrac{10}{3}\,n\,m_2$, since then $m_4\le(92n-8)m_2\le92\,n\,m_2$ (note $m_2=\Tr[K^2\rho^{\otimes2}]\ge0$).
 
\subsection{Canonical coordinates}
 
As in \SM{}~\ref{app:faf-fidelity}, a real orthogonal transformation of the Majorana basis brings the covariance matrix to the canonical form $\Gamma_\rho=\bigoplus_{j=1}^n(1-2p_j)\left(\begin{smallmatrix}0&1\\-1&0\end{smallmatrix}\right)$. We write $\widetilde\gamma_a$ for the canonical Majoranas and $c_j=(\widetilde\gamma_{2j-1}+i\widetilde\gamma_{2j})/2$ for the canonical modes, $p_j=\langle c_j^\dagger c_j\rangle_\rho\in[0,1]$ are their occupations.  All quantities below are invariant under this transformation.  From $m_2=2n-\|\Gamma_\rho\|_F^2$,
\begin{equation}
\label{eq:m2-8D}
m_2=2\sum_{j=1}^n\big(1-(1-2p_j)^2\big)=8\sum_{j=1}^n p_j(1-p_j)=2\,\FAF_1(\rho).
\end{equation}
In the canonical basis all cross-mode two-point functions vanish, normal ($\langle c_i^\dagger c_j\rangle_\rho=p_i\delta_{ij}$) and anomalous ($\langle c_ic_j\rangle_\rho=0$) alike.
 
\subsection{Upper bound for $A_4$}
 
For $z\in\mathbb C^{2n}$, we write $\gamma(z)=\sum_a z_a\widetilde\gamma_a$ and denote by $W(z,z')=\Tr[\rho\,\gamma(z)\gamma(z')]$ the two-point kernel, bilinear in both arguments. The canonical anticommutation relations give $W(z,z')+W(z',z)=2z^{\mathsf T}z'$.  We define the residual
\begin{equation}
\label{eq:kappa-def}
\kappa(z_1,z_2,z_3,z_4)=\Tr[\rho\,\gamma(z_1)\gamma(z_2)\gamma(z_3)\gamma(z_4)]-W_{12}W_{34}+W_{13}W_{24}-W_{14}W_{23},
\end{equation}
with $W_{pq}=W(z_p,z_q)$.  A direct check using the anticommutation relations shows that $\kappa$ vanishes under symmetrization in any adjacent pair of arguments. Hence, $\kappa$ is a fully antisymmetric multilinear form, completely specified by the coefficients $\kappa(e_a,e_b,e_c,e_d)$ with $a<b<c<d$ in any orthonormal basis $\{e_a\}$ of $\mathbb{C}^{2n}$, and we set $\|\kappa\|^2=\sum_{a<b<c<d}|\kappa(e_a,e_b,e_c,e_d)|^2$. This quantity is the same in every orthonormal basis: because $\kappa$ is linear, with no complex conjugation, in each of its four arguments, a change of basis $e_a\to Ue_a$ mixes the coefficients through the matrix $(U^{\mathsf T})^{\otimes 4}$, which is unitary whenever $U$ is, so the sum of squared moduli is unchanged.

For a Gaussian state Wick's theorem gives $\kappa\equiv0$, so $\|\kappa\|$ measures the failure of Wick factorization. We now show that $\kappa$ controls $A_4$.  First, since $\sum_a\widetilde\gamma_a\otimes\widetilde\gamma_a=\sum_a\gamma_a\otimes\gamma_a$ for any real orthogonal change of Majorana basis, the moments $m_2,m_4$, and hence $A_4$, which is determined by them through Eq.~\eqref{eq:exact-m4}, are the same in every Majorana basis, and we may evaluate $A_4$ in the canonical one, $A_4=\sum_{a<b<c<d}\big(\Tr[\rho\,\widetilde\gamma_a\widetilde\gamma_b\widetilde\gamma_c\widetilde\gamma_d]\big)^2$.  Second, evaluating Eq.~\eqref{eq:kappa-def} on the basis vectors $e_a$ attached to the canonical Majoranas and rearranging, each summand of $A_4$ is a coefficient of $\kappa$ up to Wick products,
\begin{equation}
\label{eq:A4-summand}
\Tr[\rho\,\widetilde\gamma_a\widetilde\gamma_b\widetilde\gamma_c\widetilde\gamma_d]
=\kappa(e_a,e_b,e_c,e_d)+W_{ab}W_{cd}-W_{ac}W_{bd}+W_{ad}W_{bc},
\qquad
W_{ab}:=W(e_a,e_b)=\Tr[\rho\,\widetilde\gamma_a\widetilde\gamma_b],
\end{equation}
and in the canonical basis $W_{ab}$ vanishes whenever the two indices belong to different mode pairs $\{2j-1,2j\}$.

To proceed, we classify the quadruples $a<b<c<d$ according to how their four indices are distributed over the mode pairs $\{1,2\},\{3,4\},\dots,\{2n-1,2n\}$.  Since each mode pair can contribute at most two indices, only three patterns occur, which we label by the number of indices drawn from each participating pair: $(2,2)$, the four indices form two complete mode pairs; $(2,1,1)$, one complete mode pair plus one index from each of two further pairs; and $(1,1,1,1)$, one index from each of four different pairs.  In the patterns $(2,1,1)$ and $(1,1,1,1)$, each of the three Wick products in Eq.~\eqref{eq:A4-summand} contains at least one factor connecting two different mode pairs and therefore vanishes, so the $A_4$ summand equals $|\kappa(e_a,e_b,e_c,e_d)|^2$.  In the pattern $(2,2)$---two complete mode pairs $\{2i-1,2i\}$ and $\{2j-1,2j\}$ with $i<j$, of which there are exactly $\binom n2$---the summand is instead bounded by $1$, because $\widetilde\gamma_{2i-1}\widetilde\gamma_{2i}\widetilde\gamma_{2j-1}\widetilde\gamma_{2j}$ is a Hermitian involution and hence has expectation in $[-1,1]$.  Therefore
\begin{equation}
\label{eq:A4-kappa}
A_4\;\le\;\tbinom n2\;+\;\sum_{(2,1,1),(1,1,1,1)}|\kappa(e_a,e_b,e_c,e_d)|^2\;\le\;\tbinom n2+\|\kappa\|^2 ,
\end{equation}
where the last step simply adds the non-negative $(2,2)$-coefficients of $\kappa$ to complete the norm.  The $\binom n2$ term is cancelled exactly by the offset in Eq.~\eqref{eq:exact-m4}, so it remains to prove $\|\kappa\|^2\le\tfrac{10}{3}\,n\,m_2$.

\subsection{A two-body estimate}
 
\begin{theorem}
\label{thm:2body}
Let $b_1,\dots,b_n$ be any family satisfying the canonical anticommutation relations on the Hilbert space of $\rho$, and set $\xi_{ij}=\Tr[\rho\,b_j^\dagger b_i]$, $\bar N=\Tr\xi$, $D_b=\Tr[\xi(\id-\xi)]$, and
\begin{equation}
\label{eq:K2-def}
K^{(2)}_{k\ell,ij}=\Tr[\rho\,b_i^\dagger b_j^\dagger b_\ell b_k]-\xi_{ki}\xi_{\ell j}+\xi_{\ell i}\xi_{kj}.
\end{equation}
Then $\|K^{(2)}\|_{\rm HS}^2\le5\,\bar N\,D_b$.
\end{theorem}
 
This extends the fixed-particle-number, pure-state estimate of Christiansen~\cite{christiansen2024hilbert} to arbitrary number-coherent mixed states.
 
\begin{proof}
(i) For an $n^2\times n^2$ matrix $A$ define the lowering maps $L_r(A)=\sum_{k,\ell,s}\bar A_{k\ell,sr}\,b_s^\dagger b_\ell b_k$ and $T_A=\sum_r\{L_r(A)^\dagger,L_r(A)\}$.  Proposition~4 of Ref.~\cite{christiansen2024hilbert} states that $P_NT_AP_N\preceq5N\|A\|_{\rm HS}^2P_N$ on each $N$-particle sector.  Each $L_r(A)$ lowers the number operator $\widehat N_b=\sum_jb_j^\dagger b_j$ by one, so $[T_A,\widehat N_b]=0$, and summing the sector inequalities gives $T_A\preceq5\|A\|_{\rm HS}^2\,\widehat N_b$ on the full Fock space; hence $\Tr[\rho\,T_A]\le5\,\bar N\,\|A\|_{\rm HS}^2$.
 
(ii)  We work in the Hilbert--Schmidt inner product $\langle R_1,R_2\rangle_2=\Tr[R_1^\dagger R_2]$, on which operators act by left multiplication (a representation of the canonical anticommutation relations), and use the cyclic vector $\rho^{1/2}$.  Then $\|C\rho^{1/2}\|_2^2=\Tr[\rho\,C^\dagger C]$ for every operator $C$, and $\Tr[\rho\{C^\dagger,C\}]=\|C\rho^{1/2}\|_2^2+\|C^\dagger\rho^{1/2}\|_2^2$.  With $b(u)=\sum_s\bar u_sb_s$ and $b^\dagger(v)=\sum_sv_sb_s^\dagger$ for $u,v\in\mathbb C^n$, one has $\|b(u)\rho^{1/2}\|_2^2=\langle u,\xi u\rangle$ and $\|b^\dagger(v)\rho^{1/2}\|_2^2=\|v\|^2-\langle v,\xi v\rangle$.
 
(iii) For all $\varphi_1,\varphi_2,\psi_1,\psi_2\in\mathbb C^n$, with $\mathcal O:=b^\dagger(\psi_1)b^\dagger(\psi_2)b(\varphi_1)$,
\begin{equation}
\label{eq:carhs}
\langle\varphi_1\otimes\varphi_2,\,K^{(2)}(\psi_1\otimes\psi_2)\rangle
=\big\langle b^\dagger(\xi\varphi_2)\rho^{1/2},\;\mathcal O\,\rho^{1/2}\big\rangle_2
-\big\langle \mathcal O^\dagger\rho^{1/2},\;b\big((\id-\xi)\varphi_2\big)\rho^{1/2}\big\rangle_2 .
\end{equation}
Indeed, anticommuting $b(\varphi_2)$ leftward through $b^\dagger(\psi_1)b^\dagger(\psi_2)$,
\begin{equation}
\Tr[\rho\,b^\dagger(\psi_1)b^\dagger(\psi_2)b(\varphi_2)b(\varphi_1)]
=\langle\varphi_2,\psi_2\rangle\langle\varphi_1,\xi\psi_1\rangle-\langle\varphi_2,\psi_1\rangle\langle\varphi_1,\xi\psi_2\rangle
+\Tr[\rho\,b(\varphi_2)\,\mathcal O],
\end{equation}
so that the left-hand side of Eq.~\eqref{eq:carhs} equals $\Tr[\rho\,b(\varphi_2)\mathcal O]+\langle\varphi_1,\xi\psi_1\rangle\langle\varphi_2,(\id-\xi)\psi_2\rangle-\langle\varphi_1,\xi\psi_2\rangle\langle\varphi_2,(\id-\xi)\psi_1\rangle$.  Splitting $b(\varphi_2)=b(\xi\varphi_2)+b((\id-\xi)\varphi_2)$ and using, with $v=(\id-\xi)\varphi_2$,
\begin{equation}
\{b(v),\mathcal O\}=\langle v,\psi_1\rangle\,b^\dagger(\psi_2)b(\varphi_1)-\langle v,\psi_2\rangle\,b^\dagger(\psi_1)b(\varphi_1),
\end{equation}
whose $\rho$-expectation exactly cancels the two residual products, one arrives at $\Tr[\rho\,b(\xi\varphi_2)\mathcal O]-\Tr[\rho\,\mathcal O\,b((\id-\xi)\varphi_2)]$, which is Eq.~\eqref{eq:carhs}.
 
(iv) Contracting Eq.~\eqref{eq:carhs} with the entries of $A^\dagger$ gives
\begin{equation}
\Tr[A^\dagger K^{(2)}]=\sum_{r=1}^n\Big[\big\langle b^\dagger(\xi e_r)\rho^{1/2},\,L_r(A^\dagger)^\dagger\rho^{1/2}\big\rangle_2-\big\langle L_r(A^\dagger)\rho^{1/2},\,b\big((\id-\xi)e_r\big)\rho^{1/2}\big\rangle_2\Big],
\end{equation}
where $\{e_r\}_{r=1}^n$ is the standard basis of $\mathbb C^n$.  Using $|\langle \Phi_1,\Phi_2\rangle-\langle \Phi_3,\Phi_4\rangle|\le(\|\Phi_1\|^2+\|\Phi_4\|^2)^{1/2}(\|\Phi_2\|^2+\|\Phi_3\|^2)^{1/2}$ for the direct sums over $r$: the first factor is $\sum_r\big(\|b^\dagger(\xi e_r)\rho^{1/2}\|_2^2+\|b((\id-\xi)e_r)\rho^{1/2}\|_2^2\big)=\Tr[\xi^2(\id-\xi)]+\Tr[\xi(\id-\xi)^2]=D_b$ by (ii), while the second factor is $\Tr[\rho\,T_{A^\dagger}]\le5\bar N\|A\|_{\rm HS}^2$ by (i).  Taking the supremum over $\|A\|_{\rm HS}=1$ yields $\|K^{(2)}\|_{\rm HS}^2\le5\bar N D_b$.
\end{proof}
 
\subsection{Closing the bound}
 
For $\sigma\in\{0,1\}^n$, we define the family of operators $b_j^{(\sigma)}=c_j$ if $\sigma_j=0$ and $b_j^{(\sigma)}=c_j^\dagger$ if $\sigma_j=1$, built from the canonical modes $c_j$.  In the canonical basis, the matrix $\xi_\sigma$ of Theorem~\ref{thm:2body} is diagonal with $j$-th entry $p_j$ (if $\sigma_j=0$) or $1-p_j$ (if $\sigma_j=1$), since all cross-mode and anomalous two-point functions vanish.  Consequently, using Eq.~\eqref{eq:m2-8D},
\begin{equation}
D_{b^{(\sigma)}}=\sum_{j=1}^np_j(1-p_j)=\frac{\FAF_1(\rho)}{4}=\frac{m_2}{8}\quad\text{for every }\sigma,\qquad \E_\sigma\bar N_\sigma=\frac n2\quad\text{for uniform i.i.d.\ }\sigma .
\end{equation}
Now, we introduce an orthonormal basis $w_j=(e_{2j-1}+ie_{2j})/\sqrt2$, $\bar w_j=(e_{2j-1}-ie_{2j})/\sqrt2$ of $\mathbb C^{2n}$, with $\{e_a\}_{a=1}^{2n}$ the standard basis attached to the canonical Majoranas, so that $\gamma(w_j)=\sqrt2\,c_j$ and $\gamma(\bar w_j)=\sqrt2\,c_j^\dagger$; the flip $\sigma_j=1$ exchanges the roles of $w_j$ and $\bar w_j$.  Writing $w_j^{\sigma},\bar w_j^{\sigma}$ for the annihilation/creation vectors of $b^{(\sigma)}$, a direct evaluation of Eq.~\eqref{eq:kappa-def} gives, for $i<j$ and $k<\ell$,
\begin{equation}
\kappa\big(\bar w_i^{\sigma},\bar w_j^{\sigma},w_\ell^{\sigma},w_k^{\sigma}\big)=4\,\big(K^{(2)}_\sigma\big)_{k\ell,ij},
\end{equation}
with $K^{(2)}_\sigma$ built from $b^{(\sigma)}$ as in Eq.~\eqref{eq:K2-def}.  We denote by $P^{(2,2)}_\sigma$ the orthogonal projection, in the wedge basis generated by $\{w_j,\bar w_j\}$, onto basis $4$-forms containing exactly two creation and two annihilation vectors of the polarization $\sigma$.  Summing squares and using the antisymmetry of $K^{(2)}_\sigma$ in both index pairs, and using Theorem~\ref{thm:2body}, we find
\begin{equation}
\big\|P^{(2,2)}_\sigma\kappa\big\|^2=4\,\big\|K^{(2)}_\sigma\big\|_{\rm HS}^2\;\le\;20\,\bar N_\sigma\,\frac{m_2}{8}.
\end{equation}

\begin{lemma}
\label{lem:vis}
$\E_\sigma P^{(2,2)}_\sigma\succeq\frac38\,\id$ on $\Lambda^4\mathbb C^{2n}$.
\end{lemma}
 
\begin{proof}
Every $P^{(2,2)}_\sigma$ is diagonal in the same fixed wedge basis, because a flip only permutes the set $\{w_j,\bar w_j\}$.  A basis $4$-form drawing its vectors from the planes $\operatorname{span}(w_j,\bar w_j)$ in pattern $(2,2)$ contains one creation and one annihilation vector per complete plane, hence exactly two creation vectors with probability $1$; in pattern $(2,1,1)$ the two singleton labels are independent fair coins, giving probability $\tfrac12$; in pattern $(1,1,1,1)$ the probability is $\binom42/2^4=\tfrac38$.  The minimum is $\tfrac38$.
\end{proof}
 
Since the $P^{(2,2)}_\sigma$ are orthogonal projections, $\|\kappa\|^2\le\frac83\,\E_\sigma\|P^{(2,2)}_\sigma\kappa\|^2\le\frac83\cdot\frac{20\,m_2}{8}\,\E_\sigma\bar N_\sigma=\frac{10}{3}\,n\,m_2$.  Together with Eqs.~\eqref{eq:exact-m4} and~\eqref{eq:A4-kappa} this gives $m_4\le(12n-8)m_2+80\,n\,m_2\le92\,n\,m_2$, i.e.\ $\E[X^2]\le46\,n\,\E[X]$, proving Theorem~\ref{thm:fourth-moment}.

\section{Global depolarizing noise and the FAF--purity witness}
\label{app:depol-witness}

We analyze the FAF--purity witness for a pure state subjected to global depolarizing noise,
\begin{equation}
\label{eq:global-depol-sm}
\rho_p
=
(1-p)\ket{\psi}\!\bra{\psi}
+
p\frac{\id}{2^n},
\qquad
0\le p\le 1 .
\end{equation}
We denote
$
x=(1-p)^2,
$
and $
r_\psi
=\|\Gamma_\psi\|_F^2/2
=
n-\FAF_1(\psi).$
We write the witness as
\begin{equation}
\label{eq:wfaf-def-sm}
W_{\FAF}(\rho)
=
\FAF_1(\rho)
-
2n\left(1-\Tr(\rho^2)^{1/n}\right).
\end{equation}

First, the purity of \(\rho_p\) is
\begin{align}
\Tr(\rho_p^2)
=
(1-p)^2
+
2(1-p)p\,2^{-n}
+
p^2\,2^{-n}  
=
x+\frac{1-x}{2^n}
=
\frac{1+(2^n-1)x}{2^n}.
\label{eq:depol-purity-sm}
\end{align}
Since the maximally mixed state has zero covariance and covariance is linear in the state,
$
\Gamma_{\rho_p}
=
(1-p)\Gamma_\psi.
$
Therefore
\begin{equation}
\FAF_1(\rho_p)
=
n-(1-p)^2\frac{\|\Gamma_\psi\|_F^2}{2}
=
n-xr_\psi .
\label{eq:depol-faf-sm}
\end{equation}
Moreover,
\begin{equation}
\Tr(\rho_p^2)^{1/n}
=
\frac{\left(1+(2^n-1)x\right)^{1/n}}{2}.
\end{equation}
Combining this with \eqref{eq:wfaf-def-sm} and \eqref{eq:depol-faf-sm}, we obtain the exact expression
\begin{equation}
\label{eq:depol-witness-sm}
W_{\FAF}(\rho_p)
=
n\left(1+(2^n-1)x\right)^{1/n}-n-xr_\psi .
\end{equation}

This expression admits a simple lower bound. For \(0\le x\le1\),
\begin{equation}
\label{eq:chord-bound-sm}
(1+x)^n
\le
1+(2^n-1)x .
\end{equation}
Indeed, this is the statement that the graph of the convex function \(f(x)=(1+x)^n\) lies below the chord joining its values at \(x=0\) and \(x=1\). Taking \(n\)-th roots in \eqref{eq:chord-bound-sm} gives
\[
\left(1+(2^n-1)x\right)^{1/n}
\ge
1+x .
\]
Substituting this into \eqref{eq:depol-witness-sm} yields
\begin{equation}
\label{eq:depol-lower-sm}
W_{\FAF}(\rho_p)
\ge
x(n-r_\psi)
=
(1-p)^2\FAF_1(\psi).
\end{equation}
Consequently, whenever the pure input state has \(\FAF_1(\psi)>0\), the ideal witness remains positive for every \(p<1\). In particular, for pure non-Gaussian states that are detected by \(\FAF_1\), global depolarizing noise does not introduce a finite ideal-witness threshold; the witness can only vanish at the maximally mixed endpoint \(p=1\).

We next evaluate the typical size of the signal for Haar-random pure states. We fix \(D=2^n\) and consider the Majorana bilinears
\(
B_{ab}=-\ii\gamma_a\gamma_b,
\)
for $a<b$.
Each such \(B_{ab}\) is traceless and satisfies \(B_{ab}^2=\id\). For a Haar-random pure state in the full Fock space,
\begin{equation}
\mathbb E_\psi
\left[
\left|
\bra{\psi}B_{ab}\ket{\psi}
\right|^2
\right]
=
\frac{\Tr(B_{ab}^2)}{D(D+1)}
=
\frac{1}{D+1}.
\end{equation}
Since
$r_\psi
=
\sum_{a<b}
\left|
\bra{\psi}B_{ab}\ket{\psi}
\right|^2$
and there are
\(
\binom{2n}{2}=n(2n-1)
\)
quadratic Majorana observables, we obtain for Haar random states
\begin{equation}
\label{eq:haar-faf-expectation-sm}
\mathbb E_\psi\,r_\psi
=
n\frac{2n-1}{2^n+1},
\qquad
\mathbb E_\psi\,\FAF_1(\psi)
=
n-n\frac{2n-1}{2^n+1}.
\end{equation}
Thus Haar-random states are almost maximally FAF-non-Gaussian~\cite{faf2025}.
For such states, \(r_\psi=O(n^2 2^{-n})\) typically. 
Near the maximally mixed endpoint, write \(1-p=\alpha\), so \(x=\alpha^2\).
In the linear-response regime \((2^n-1)\alpha^2\ll1\), expanding
\eqref{eq:depol-witness-sm} gives
\begin{equation}
    W_{\FAF}(\rho_p)
    =
    \alpha^2\left[(2^n-1)-r_\psi\right]
    +
    O\!\left((2^n\alpha^2)^2\right).
\end{equation}
For Haar-random \(\psi\), where \(r_\psi=O(n^2 2^{-n})\) typically, this becomes
\begin{equation}
 \label{eq:depol-small-signal-sm}
    W_{\FAF}(\rho_p)
    \approx
    2^n(1-p)^2
\end{equation}
within this linear-response regime.
Hence the ideal witness remains positive for every \(p<1\), but finite-sample detection close to \(p=1\) requires resolving a signal of order \eqref{eq:depol-small-signal-sm}.

Finally, we compare this witness threshold with exact mixed-state Gaussianity. For \(0<p<1\), the spectrum of \(\rho_p\) consists of one eigenvalue
\begin{equation}
\lambda_1
=
1-p+\frac{p}{2^n},
\end{equation}
and one eigenvalue
\begin{equation}
\lambda_0
=
\frac{p}{2^n}
\end{equation}
with multiplicity \(2^n-1\). On the other hand, a full-rank FGS is unitarily equivalent, by a Gaussian change of modes, to a product thermal state
\begin{equation}
\rho_G
=
\bigotimes_{j=1}^n
\begin{pmatrix}
1-q_j & 0 \\
0 & q_j
\end{pmatrix}.
\end{equation}
Its eigenvalues are
\begin{equation}
\lambda_{\mathbf x}
=
\prod_{j=1}^n
q_j^{x_j}(1-q_j)^{1-x_j},
\qquad
\mathbf x\in\{0,1\}^n .
\end{equation}
Such a product spectrum can have the multiplicity pattern \(1,2^n-1\) only when all \(q_j=1/2\), in which case all eigenvalues are equal. Therefore, for \(n\ge2\), a state of the form \eqref{eq:global-depol-sm} is Gaussian for \(0<p<1\) only in the maximally mixed case, which occurs exactly at \(p=1\). At \(p=0\), Gaussianity is equivalent to \(\ket{\psi}\) being a pure Gaussian state.

Thus, for Haar-random \(\ket{\psi}\), the depolarized state \(\rho_p\) is non-Gaussian for every \(p<1\) and becomes Gaussian at \(p=1\). The ideal FAF--purity witness has the same endpoint threshold, although the detectable signal can become small near \(p=1\).

\section{Non-Gaussianity witness in noisy matchgate circuits}
\label{app:numerics}

\begin{figure*}[htpb]
	\centering	
\includegraphics[width=0.99\textwidth]{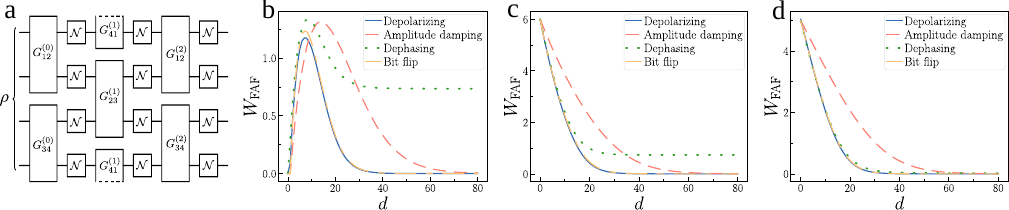}
\vspace{-0.3cm}
\caption{
Fermionic non-Gaussianity witness $W_{\mathrm{FAF}}$ in noisy matchgate (MG) circuits.  \idg{a} Circuit sketch for $n=4$ qubits and depth $d=3$, initialized in $\rho$.  Each layer consists of random matchgates $G_{ij}$ followed by local noise  $\mathcal N$.  \idg{b}--\idg{d} Witness dynamics after $d$ layers for the initial states  \idg{b} $\ket{0}^{\otimes n}$, \idg{c} the GHZ state, and \idg{d} Haar-random states.  We consider local depolarizing noise, amplitude damping, dephasing, and bit-flip noise. Data are averaged over 20 random circuit instances.}
	\label{fig:witness}
\end{figure*}

Here, we perform numerical experiments to probe the FAF--purity witness under noisy free-fermion dynamics.  The goal is to study a simple setting in which the coherent part of the dynamics is free-fermionic, while the noise can drive the state away from the mixed Gaussian manifold.  We consider an \(n\)-qubit system and track the evolution of the witness $W_{\FAF}(\rho)$ as a function of circuit depth $d$.

\prlsection{Initial states.}
We use three representative classes of pure initial states.  The computational basis states $\ket{0}^{\otimes n}$ and $\ket{1}^{\otimes n}$ are pure FGS; initially they satisfy $\FAF_1=W_{\FAF}=0$.  Since matchgate unitaries preserve Gaussianity, any nonzero value of  $W_{\FAF}$ generated from these states is caused by the noise channel together with the intervening Gaussian scrambling, rather than by the coherent dynamics alone.  As non-Gaussian initial states we use the parity-cat state $(\ket{0}^{\otimes n}+\ket{1}^{\otimes n})/\sqrt2$ and Haar-random pure states.  

\prlsection{Noisy matchgate circuit.}
The coherent part of the evolution is a random nearest-neighbor matchgate circuit with periodic boundary conditions.  We use a staggered brickwork pattern, as shown in Fig.~\ref{fig:witness}a, alternating between the pairs $(0,1),(2,3),\ldots$ and $(1,2),(3,4),\ldots$.  On each pair $(i,j)$ we draw an independent two-mode matchgate $G_{ij}=\exp(-iH_{ij})$, where $H_{ij}=(g/\sqrt6)\sum_{a<b\in\mathcal M_{ij}}\theta_{ab} (-i\gamma_a\gamma_b)$, $\mathcal M_{ij}=\{2i-1,2i,2j-1,2j\}$, $\theta_{ab}\sim\mathcal N(0,1)$, and we set $g=1$.  The factor $\sqrt6$ normalizes over the six quadratic Majorana bilinears on two modes.
In the absence of noise, the full depth-$d$ circuit is a fermionic Gaussian unitary that preserves $\mathcal{G}_n$.

\prlsection{Noise models.}
After each matchgate layer we apply an identical one-mode channel
$\mathcal N_p$ to every mode,
\(
    \rho_{\ell+1}
    =
    \mathcal N_p^{\otimes n}
    \!\left(U_\ell\rho_\ell U_\ell^\dagger\right).
\)
In the Jordan--Wigner representation we consider four standard local qubit channels: depolarizing noise, amplitude damping, dephasing, and bit-flip noise.  Depolarizing noise is unital and drives the state toward the maximally mixed Gaussian state.  Amplitude damping is non-unital and drives the system toward the Gaussian vacuum.  Dephasing corresponds to random local parity operations; each branch maps Gaussian states to Gaussian states, but their mixture need not be Gaussian because the mixed Gaussian set is not convex.  Bit-flip noise changes local occupations and, when combined with matchgate scrambling, tends to wash the state toward high-entropy occupation-basis mixtures.

\prlsection{Results.}
The results are shown in Fig.~\ref{fig:witness}.  For Gaussian occupation states, $W_{\FAF}$ is initially zero.  Noise can nevertheless generate positive values of $W_{\FAF}$  after matchgate scrambling, followed by relaxation at larger depths.  The late-time behavior depends strongly on the noise channel:
dephasing can stabilize a finite positive witness, consistent with
non-Gaussian mixtures of Gaussian branches, whereas depolarizing noise and bit-flip noise tend to suppress the witness toward a nearly Gaussian high-entropy state.  Amplitude damping instead relaxes the system toward the Gaussian vacuum.

For initially non-Gaussian states, such as the parity-cat and Haar-random states, the witness typically decreases with depth as local noise degrades the coherent non-Gaussian structure.  The limiting value is not universal and depends on both the initial state and the noise channel.  In particular, dephasing can leave a finite witness for some initial states, whereas the unital and dissipative channels generally drive $W_{\FAF}$ close to zero.  A negative value of $W_{\FAF}$ should not be interpreted as evidence of Gaussianity; it only indicates that this sufficient witness no longer
certifies non-Gaussianity.

\section{Quantum pseudorandomness and pseudo non-Gaussianity}
\label{app:pseudo-nongaussianity}
First, let us recall the subset phase state~\cite{aaronson2022quantum}
\begin{equation}
\label{eq:subset-phase-supp}
    \ket{\psi_{r,S}}
    =
    \frac{1}{\sqrt{M}}
    \sum_{x\in S}
    (-1)^{r(x)}
    \ket{x},
    \qquad
    M=|S|=2^q ,
\end{equation}
where $S\subseteq\{0,1\}^n$ and $r:S\to\{0,1\}$ is a phase function. 

\subsection{FAF of random subset phase states}

\begin{proposition}[FAF of random subset phase states]
\label{prop:pseudorandom-faf}
Let $S$ be a uniformly random subset of $\{0,1\}^n$ of size $M$ and let $r$ be uniformly random, independent of $S$. Then
\begin{equation}
\label{eq:subset-phase-faf-bound-main}
\mathbb E_{S,r}\!\left[\FAF_1(\psi_{r,S})\right]
\ge
n-\frac{n(2n-1)}{M}.
\end{equation}
Consequently, for every $\eta>0$,
\begin{equation}
\Pr_{S,r}\!\left[\FAF_1(\psi_{r,S})\le n-\eta\right]
\le
\frac{n(2n-1)}{M\eta}.
\end{equation}
In particular, if $M=2^q\gg n$, then $\FAF_1(\psi_{r,S})=n-o(n)$ with high probability.
\end{proposition}
\begin{proof}

We fix $\ket{\psi_{r,S}}$ via Eq.~\eqref{eq:subset-phase-supp}.
For $a<b$, we consider the Hermitian Majorana bilinear
$
B_{ab}=-\ii\gamma_a\gamma_b,
$
so that $\Gamma_{ab}=\bra{\psi}B_{ab}\ket{\psi}$. The normalized covariance weight is
\begin{equation}
n-\FAF_1(\psi)=\frac{\|\Gamma_\psi\|_F^2}{2}
=\sum_{a<b}|\bra{\psi}B_{ab}\ket{\psi}|^2.
\label{eq:cov-weight-subset-proof}
\end{equation}
There are $n(2n-1)$ terms in the sum.

We first bound the diagonal mode bilinears. For each mode $j$, $B_{2j-1,2j}$ is, up to sign, the computational-basis Pauli $Z_j$. Hence
\[
\bra{\psi_{r,S}}B_{2j-1,2j}\ket{\psi_{r,S}}
=\pm\frac1M\sum_{x\in S}(-1)^{x_j}.
\]
If $S$ is a uniformly random subset of size $M$ from $N=2^n$ bit strings, then sampling without replacement gives
\begin{equation}
\mathbb E_S\left|\frac1M\sum_{x\in S}(-1)^{x_j}\right|^2
=
\frac{N-M}{M(N-1)}
\le \frac1M.
\label{eq:diag-subset-bound}
\end{equation}
This bound is independent of the phases.

For a pair of distinct modes \(j<k\), the four Majorana bilinears connecting these modes have the same bit-flip mask \(s_{jk}\), which flips the occupations of modes \(j\) and \(k\). Since \(\ket{\psi_{r,S}}\) has real amplitudes in the computational basis, the two corresponding Pauli strings with an odd number of
\(Y\)'s have zero expectation value. It remains  to bound the two real strings.

Fix one of these real bilinears and write \(A_{jk}=S\cap(S\oplus s_{jk})\). Partition \(A_{jk}\) into unordered pairs \(\{x,x\oplus s_{jk}\}\). For each such pair, the two directed contributions carry the same random sign \((-1)^{r(x)+r(x\oplus s_{jk})}\). Averaging over the random phases removes cross terms between distinct unordered pairs, and each pair contributes at most \(4\). Hence, conditional on \(S\),
\begin{equation}
    \mathbb E_r
    \left|
    \bra{\psi_{r,S}}B_{ab}\ket{\psi_{r,S}}
    \right|^2
    \le
    \frac{2|A_{jk}|}{M^2}
    \le
    \frac{2}{M}.
\end{equation}
Thus the four off-diagonal bilinears associated with a fixed pair of modes contribute at most \(4/M\) in total. Summing over the \(\binom n2\) mode pairs gives an off-diagonal contribution bounded by \(2n(n-1)/M\). Together with the diagonal contribution \(n/M\), this yields
\begin{equation}
    \mathbb E_{S,r}[n-\FAF_1(\psi_{r,S})]
    \le
    \frac{n(2n-1)}{M}.
\end{equation}
\end{proof}

\subsection{Quantum pseudorandom states}

\begin{definition}[Quantum pseudorandom state ensemble~\cite{ji2018pseudorandom}]
\label{def:prs}
An efficiently preparable ensemble of \(n\)-qubit pure states
\[
\mathcal E_n=\{|\psi_k\rangle:k\in\mathcal K_n\}
\]
with key $k$ from keyspace $\mathcal K=\{0,1\}^{\text{poly}(n)}$ is a quantum pseudorandom state ensemble if, for every quantum polynomial-time
algorithm \(\mathcal D\) and every polynomially bounded number of copies
\(t=\text{poly}(n)\),
\[
\left|
\Pr_{k\leftarrow\mathcal K_n}
\!\left[
\mathcal D\!\left(|\psi_k\rangle^{\otimes t}\right)=1
\right]
-
\Pr_{|\phi\rangle\leftarrow{\rm Haar}}
\!\left[
\mathcal D\!\left(|\phi\rangle^{\otimes t}\right)=1
\right]
\right|
\le
{\rm negl}(n).
\]
Here \({\rm negl}(n)\) denotes a function smaller than \(1/p(n)\) for every
polynomial \(p\), and \(|\phi\rangle\leftarrow{\rm Haar}\) denotes a Haar-random
pure state on Hilbert space \(\mathcal H_N\).
\end{definition}

Equivalently, no efficient observer with access to polynomially many copies
can distinguish the pseudorandom ensemble from the Haar ensemble with more than
negligible advantage.  The definition is operational: any property that can be
estimated efficiently from polynomially many copies and that has a sharply
different value on Haar-random states must also hold, up to negligible error,
for pseudorandom states.

\subsection{Pseudo non-Gaussianity gaps}

The two-copy Bell/FAF measurement estimates \(\FAF_1\) efficiently.  Therefore,
if a pseudorandom ensemble had noticeably smaller \(\FAF_1\) than Haar-random
states, the Bell/FAF measurement would give an efficient distinguisher.  This
motivates the following definition, following similar definitions for other quantum resources~\cite{aaronson2022quantum,haug2025pseudorandom,gu2023pseudomagic}:

\begin{definition}[Pseudo non-Gaussianity]
\label{def:pseudo-nongaussianity-gap}
Let \(f(n)\) and \(g(n)\) be two functions with \(0\le g(n)<f(n)\le n\).  Two
efficiently preparable ensembles $\{\psi_k\}_k$, $\{\phi_k\}_k$ exhibit a pseudo
non-Gaussianity gap \((f_n,g_n)\) with respect to a non-Gaussianity measure
\(\mathcal M\) if:
\begin{enumerate}
    \item Every state in the ensemble can be prepared by a polynomial-size
    quantum circuit
    \item The two ensembles are computationally indistinguishable from Haar-random
    states using polynomially many copies
    \item States in the first ensemble have large non-Gaussianity,
    \[
    \mathcal M(\psi_k)= f(n)
    \]
    with overwhelming probability, whereas the second ensemble has low non-Gaussianity
    \[
    \mathcal M(\phi_k)= g(n) .
    \]
\end{enumerate}

\end{definition}
As FAF can be efficiently estimated to $1/\text{poly}(n)$ precision, this forces the low non-Gaussianity ensemble to be 
\[
g(n)=\mathbb E[\FAF_{1}(\phi_k)]_k=n-\text{negl}(n),
\]
for all but a negligible fraction of keys.  Otherwise, estimating \(\FAF_1\)
with the Bell/FAF measurement to inverse-polynomial precision would separate
this  ensemble from a highly non-Gaussian ensemble such as Haar-random states, which have (see~\eqref{eq:haar-faf-expectation-sm})
\[
f(n)=\mathbb E[\FAF_{1}(\psi_k)]_k=n-\Theta(2^{-n}).
\]
Instead of Haar random states, the highly non-Gaussian ensemble can also be efficiently preparable subset phase states with $M =2^n$, which have similar non-Gaussianity as Haar random states.

\subsection{Preparation lower bound from local non-Gaussian gates}

The FAF also gives a simple circuit lower bound in a
restricted preparation model.  Consider circuits that start from a fermionic
Gaussian state and interleave arbitrary fermionic Gaussian unitaries with
\(t\) non-Gaussian gates, each supported on at most \(m\) fermionic modes.
We recall that Gaussian unitaries preserve \(\FAF_{1}\).

A gate supported on \(m\) modes can change only covariance entries involving
those \(m\) modes.  
Since any physical covariance matrix has singular values at
most one, the Frobenius weight of the affected covariance rows and columns is
\(O(m)\).  

More formally, let $A\subset\{1,\ldots,2n\}$ be the set of Majorana indices belonging to a subset of $m$ fermionic modes and $\vert A\vert=2m$. Let $B=A^c$ be the complementary set of Majorana indices. We write the covariance matrix for any intermediate state $\rho$ in block form as
\[
\Gamma
=
\begin{pmatrix}
\Gamma_{AA} & \Gamma_{AB}\\
-\Gamma_{AB}^T & \Gamma_{BB}
\end{pmatrix},
\]
and define the covariance weight involving the subsystem $A$ by
\[
S_A(\Gamma)
:=
\sum_{\substack{a<b\\ a\in A\ \mathrm{or}\ b\in A}}
\Gamma_{ab}^2 =
\sum_{\substack{a<b\\ a,b\in A}}\Gamma_{ab}^2
+
\sum_{\substack{a\in A\\ b\in B}}\Gamma_{ab}^2
=
\frac12\|\Gamma_{AA}\|_F^2+\|\Gamma_{AB}\|_F^2 .
\]
We first record a simple bound on $S_A(\Gamma)$. For any physical covariance matrix, all singular values of $\Gamma$ are bounded by one, implying $\Gamma\Gamma^T\le I_{2n}$.
Let $P_A$ denote the projector onto the Majorana indices in $A$. Then
\[
\|P_A\Gamma\|_F^2=\mathrm{tr}(P_A\Gamma\Gamma^TP_A)\le\mathrm{tr}(P_A)
=|A|=2m .
\]
On the other hand,
\[
\|P_A\Gamma\|_F^2
=\|\Gamma_{AA}\|_F^2+\|\Gamma_{AB}\|_F^2
=2\sum_{\substack{a<b\\a,b\in A}}\Gamma_{ab}^2+\sum_{\substack{a\in A\\b\in B}}\Gamma_{ab}^2 .
\]
Since all terms are nonnegative, this implies
\[
S_A(\Gamma)
=
\sum_{\substack{a<b\\a,b\in A}}\Gamma_{ab}^2
+
\sum_{\substack{a\in A\\b\in B}}\Gamma_{ab}^2
\le
\|P_A\Gamma\|_F^2
\le
2m .
\]
Now let $U_A$ be an arbitrary unitary supported only on the $m$ fermionic modes corresponding to $A$, and define
\[
\rho' = U_A\rho U_A^\dagger,
\qquad
\Gamma'=\Gamma_{\rho'} .
\]
Since $U_A$ acts trivially on the complement $B$, every bilinear supported entirely on $B$ is unchanged and thus $\Gamma'_{BB}=\Gamma_{BB}$.
Therefore the change in $\mathrm{FAF}_1$ can only come from covariance entries involving at least one Majorana index in $A$. 
From the definition of $\FAF_1$, we obtain
\[
\mathrm{FAF}_1(\rho')
-
\mathrm{FAF}_1(\rho)
=
-\Bigl(S_A(\Gamma')-S_A(\Gamma)\Bigr).
\]
Thus
\[
\left|
\mathrm{FAF}_1(\rho')
-
\mathrm{FAF}_1(\rho)
\right|
\le
S_A(\Gamma')+S_A(\Gamma).
\]
Applying the bound $S_A(\Gamma)\le 2m$ to both $\Gamma$ and $\Gamma'$ gives
\begin{equation}
\left|\mathrm{FAF}_1(U_A\rho U_A^\dagger)
-\mathrm{FAF}_1(\rho)\right|\le4m.
\label{eq:fafLocChange}
\end{equation}
Therefore an arbitrary unitary acting on $m$ fermionic modes can change $\mathrm{FAF}_1$ by at most $O(m)$.
In passing, we note that the inequality~\eqref{eq:fafLocChange} generalizes a similar result in Sec.~IVa of \cite{faf2025} that was specialized to $\rho \in \mathcal{G}_n$.

We can now apply above considerations to circuits made from Gaussian unitaries and local non-Gaussian gates. Fermionic Gaussian unitaries 
preserve $\mathrm{FAF}_1$:
\(
\mathrm{FAF}_1(U_G\rho U_G^\dagger)
=
\mathrm{FAF}_1(\rho).
\)
Suppose a circuit starts from a pure fermionic Gaussian state $\rho_0$, so that
\[
\mathrm{FAF}_1(\rho_0)=0,
\]
and then applies arbitrary Gaussian unitaries interleaved with $t$ non-Gaussian gates, each supported on at most $m$ fermionic modes. Since Gaussian gates do not change $\mathrm{FAF}_1$, and each local non-Gaussian gate can increase $\mathrm{FAF}_1$ by at most $4m$, we have
\[
\mathrm{FAF}_1(\rho_{\mathrm{out}})
\le
4mt .
\]
Consequently, if a target family of states satisfies
\[
\mathrm{FAF}_1(\rho_{\mathrm{out}})\ge c n
\]
for some constant $c>0$, then any such preparation circuit must contain at least
\(
t\ge \frac{c}{4}\frac{n}{m}
\)
local non-Gaussian gates. In particular, for constant-size non-Gaussian gates, $m=O(1)$, this implies the bound
\[
t=\Omega(n).
\]
Combining this with the minimal pseudoresource gap $g(n)$, any
pseudorandom state ensemble generated in this model must satisfy
\(
O(mt)\ge n-\text{negl}(n),
\)
and hence
\(
t=\Omega(n/m).
\)
In particular, if the non-Gaussian gates have constant support with $m=O(1)$, then the count of non-Gaussian gates to prepare a pseudo non-Gaussian ensemble is
\(
t=\Omega(n).
\)

\section{Fock-basis coherence from Bell measurement}
\label{app:coherence}

Here, we show that addition to fermionic non-Gaussianity, one can use the same two-copy Bell measurement data to estimate a simple quadratic coherence functional in the Fock basis.
Here coherence~\cite{streltsov2017colloquium} is defined with respect to the Fock (or computational) basis
\(
\{|x\rangle:x\in\{0,1\}^n\}.
\)
Let \(\Delta\) denote complete dephasing in this basis,
\[
\Delta(\rho)
=
\sum_{x\in\{0,1\}^n}
|x\rangle\!\langle x|\rho|x\rangle\!\langle x|.
\]
A natural quadratic coherence measure is
\begin{equation}
\label{eq:C2-coherence-def}
C_2(\rho)
=
\Tr(\rho^2)-\Tr(\Delta(\rho)^2).
\end{equation}
This quantity is nonnegative, basis dependent, and vanishes for every Fock-basis diagonal
state.  For a pure state
\(
|\psi\rangle
=
\sum_x \psi_x |x\rangle ,
\)
one has \(\Tr(|\psi\rangle\!\langle\psi|^2)=1\), and therefore
\begin{equation}
\label{eq:C2-pure}
C_2(\psi)
=
1-
\sum_x |\psi_x|^4 .
\end{equation}
Thus \(C_2(\psi)\) is one minus the collision probability of the Fock-basis measurement distribution~\cite{Dalzell22anti, lami2025,Turkeshi24hilbert,p8dn-glcw,lkwg-4dbt,aditya2025growthspreadingquantumresources,magni2025anticoncentrationstatedesigndoped,clayes}.

Consider two copies of a pure state,
\[
|\psi\rangle\otimes|\psi\rangle .
\]
The Bell measurement is implemented by applying,
for each mode \(j=1,\ldots,n\), a CNOT gate from copy \(A\) to copy \(B\),
followed by a Hadamard gate on copy \(A\), and then measuring both copies in the occupation basis.  Let
$
u\in\{0,1\}^n
$
denote the measured bit string on the first copy and
$
v\in\{0,1\}^n
$
the measured bit string on the second copy after the Bell basis change.

The second-register string \(v\) records the bitwise difference between the two
pre-Bell occupation strings.  In particular,
$
v=0^n
$
occurs exactly when the two original occupation-basis samples coincide.  More
explicitly, writing
$
|\psi\rangle=\sum_x \psi_x |x\rangle ,
$
the two-copy state before the Bell basis change is
\[
|\psi\rangle^{\otimes2}
=
\sum_{x,y}\psi_x\psi_y\,|x\rangle_A|y\rangle_B .
\]
After the CNOTs \(A\to B\), the second register contains \(x\oplus y\):
\[
|x\rangle_A|y\rangle_B
\mapsto
|x\rangle_A|x\oplus y\rangle_B .
\]
Hence the event \(v=0^n\) receives contributions only from \(x=y\).  The
subsequent Hadamards on the first register do not change the marginal
probability of the second register.  Therefore
\begin{equation}
\label{eq:bell-vzero-collision}
\Pr[v=0^n]
=
\sum_x |\psi_x|^4
=
\Tr(\Delta(|\psi\rangle\!\langle\psi|)^2).
\end{equation}
The same reasoning has a direct mixed-state extension. For an arbitrary mixed state \(\rho\), the same event estimates the dephased purity:
\begin{equation}
    \Pr[v=0^n]
    =
    \Tr(\Delta(\rho)^2).
\end{equation}
Indeed, the condition \(v=0^n\) projects the two copies onto equal computational-basis strings. The same Bell outcome also gives the eigenvalue of the global swap operator,
\begin{equation}
    X_{\rm pur}
    =
    (-1)^{\sum_{j=1}^n u_jv_j},
    \qquad
    \mathbb E[X_{\rm pur}]
    =
    \Tr(\rho^2).
\end{equation}
Therefore,
\begin{equation}
    X_{\rm pur}-\mathbf 1[v=0^n]
\end{equation}
is an unbiased single-shot estimator of the quadratic Fock-basis coherence
\begin{equation}
    C_2(\rho)
    =
    \Tr(\rho^2)-\Tr(\Delta(\rho)^2).
\end{equation}
For pure states, \(X_{\rm pur}=1\) deterministically, and the estimator reduces
to the Bernoulli variable
\begin{equation}
    Y=\mathbf 1[v\neq0^n],
    \qquad
    \mathbb E[Y]=C_2(\psi).
\end{equation}
This estimator is the destructive Bell-measurement analogue of the \(q=2\) computational-basis inverse participation ratio measurement of Ref.~\cite{Liu25parti}: the event \(v=0^n\) directly estimates the collision probability \(I_2^Z=\sum_x|\psi_x|^4\), so that \(C_2(\psi)=1-I_2^Z(\psi)\).

For pure states, let \(Y_1,\ldots,Y_N\) be the outcomes from \(N\) independent
Bell shots, with \(Y_r=\mathbf 1[v_r\neq0^n]\). The estimator
\begin{equation}
    \widehat C_2
    =
    \frac1N\sum_{r=1}^N Y_r
\end{equation}
is unbiased. Since each \(Y_r\in\{0,1\}\), Hoeffding's inequality gives
\begin{equation}
    \Pr\!\left[
    |\widehat C_2-C_2(\psi)|\ge \eta
    \right]
    \le
    2\exp(-2N\eta^2).
\end{equation}
Consequently, \(N\ge (2\eta^2)^{-1}\log(2/\delta)\) Bell shots suffice to
estimate \(C_2(\psi)\) to additive accuracy \(\eta\) with failure probability at
most \(\delta\).
Thus the sample complexity is
\[
N
=
O\!\left(
\eta^{-2}\log\frac1\delta
\right),
\]
i.e. is independent of qubit number $n$.

\section{Other Gaussianity tests}\label{app:tests}

In this section, we compare different single and multi-copy  fermionic non-Gaussianity tests from the literature.

\subsection{Random purification test}
Ref.~\cite{walter2025random} introduced the random-purification-based Gaussianity test with sample complexity $O(n^2/\epsilon^2)$. The central idea of this method is to reduce mixed-state fermionic Gaussian tomography and testing to the pure-state case: given many copies of a mixed state, a random purification channel produces copies of a randomly chosen purification which is itself fermionic Gaussian whenever the input state is fermionic Gaussian. One can then apply pure-state Gaussian tomography or testing to the purified state. Both learning and testing scales as $O(n^2)$. 
However, the protocol requires implementing the fermionic random purification channel over $O(n^2)$ copies of the state, which does not seem to have a direct Clifford implementation.

\subsection{Matchgate shadow and joint measurement tests}

One can test fermionic non-Gaussianity by first learning the covariance matrix and then evaluating $\FAF_1$. For pure states, Lemma~\ref{lem:distance-faf} implies that every state with $\epsilon_G(\psi)\ge\epsilon$ satisfies $\FAF_1(\psi)\ge 2\epsilon^2$.  It is therefore enough to estimate $\FAF_1$ to additive precision $\eta=\Theta(\epsilon^2)$. Matchgate-shadow measurements estimate quadratic Majorana correlators with variance parameter \(O(n)\). Hence all \(O(n^2)\) correlators \(\mu_{ab}=\langle B_{ab}\rangle\) can be learned to additive accuracy \(\alpha\), uniformly over \(a<b\), using \(O(n\alpha^{-2}\log(n/\delta))\) samples. Denoting the corresponding estimates by \(\widehat\mu_{ab}\), we have, on this high-probability event,
\[
    \left|
        \sum_{a<b}\widehat\mu_{ab}^{\,2}
        -
        \sum_{a<b}\mu_{ab}^2
    \right|
    \le
    O(n^{3/2}\alpha+n^2\alpha^2),
\]
where we used \(\sum_{a<b}\mu_{ab}^2\le n\).
Hence, it suffices to take
$\alpha=\Theta(\eta/n^{3/2})$.  This gives
\[
    N_{\rm MG}
    =
    O\!\left(
        n^4\eta^{-2}\log(n/\delta)
    \right)
\]
samples for additive estimation of $\FAF_1$. 
We note that for the joint measurement test, Ref.~\cite{majsak2025simple} notes that one achieves the same measurement complexity as matchgate shadows. 

Then, from the Lemma~\ref{lem:distance-faf} bounding the Gaussian fidelity and $\FAF_1$, we get
\[
    N_{\rm MG}^{\rm test}
    =
    O\!\left(
        n^4\epsilon^{-4}\log(n/\delta)
    \right)
\]
for pure-state Gaussianity testing for both matchgate shadow and joint measurement testers.

\subsection{Covariance tomography test}

A natural covariance-matrix spectrum baseline, closely related to the covariance-matrix approach of Bittel \emph{et al.}~\cite{bittel2025optimal}; see also Ref.~\cite{mele2025efficient}, is to estimate the full covariance matrix, compute its normal eigenvalues, and accept only if all of them are close to one.
Using commuting measurements for the covariance matrix, this gives a pure-state Gaussianity tester with sample complexity $O(n^5\epsilon^{-4}\log(n/\delta))$. 
The worse scaling compared with the \(\FAF_1\)-based estimators comes from resolving each covariance singular value to accuracy \(O(\epsilon^2/n)\). By contrast, \(\FAF_1=\sum_j(1-\nu_j^2)\) is a single scalar gap and only needs to be estimated to accuracy \(\Theta(\epsilon^2)\).

\subsection{Randomized FAF estimation protocol}
Here, for comparison with the single-copy $\FAF_1$ estimator from Theorem~\ref{thm:single-copy-tester}, we
consider the most direct randomized estimator of
\[
    \FAF_1(\rho)
    =
    n-\sum_{e\in\mathcal E}\langle B_e\rangle_\rho^2,
    \qquad
    \mathcal E=\{(a,b):1\le a<b\le 2n\}.
\]
Let $N_b:=|\mathcal E|=n(2n-1)$.  In one trial, choose
$e\in\mathcal E$ uniformly at random, measure $B_e$ on two independent
preparations of $\rho$, and multiply the outcomes $X,Y\in\{\pm1\}$.  Since
the two measurements are independent,
$
    \mathbb E[XY\mid e]=\langle B_e\rangle_\rho^2
$.
Therefore the rescaled variable $Z:=N_bXY$ satisfies
\(
    \mathbb E[Z]
    =
    \sum_{e\in\mathcal E}\langle B_e\rangle_\rho^2 ,
\)
and hence
\(
    \widehat{\FAF}_1
    =
    n-\frac1N\sum_{r=1}^N Z_r
\)
is an unbiased estimator of $\FAF_1(\rho)$.

This estimator is simple but inefficient: each trial probes only one of the
$N_b=O(n^2)$ bilinears, and $|Z|=N_b$.  Thus
$\operatorname{Var}(Z)\le N_b^2=O(n^4)$, so additive root-mean-square
accuracy $\eta$ requires $N=O(n^4\eta^{-2})$ trials.  With standard
success-probability amplification this becomes
$O(n^4\eta^{-2}\log(1/\delta))$.  For pure-state testing,
$\eta=\Theta(\epsilon^2)$, giving
\[
    N
    =
    O\!\left(
        n^4\epsilon^{-4}\log(1/\delta)
    \right).
\]
The matching estimator of Theorem~\ref{thm:single-copy-tester} improves the leading
scaling by one power of $n$ by measuring $n$ commuting bilinears in each
setting, rather than sampling them one at a time.

\subsection{Convolution test}
The convolution test of Ref.~\cite{lyu2024convolution,coffman2025measuring} applies the balanced fermionic convolution to two copies and SWAP-tests the result against a third copy. The test is shown in Fig.~\ref{fig:convolution-circuit-paper}.
Its rejection probability is given by
\begin{equation}
q_{\boxplus}(\psi)
=
\frac12\left[
1-\Tr\bigl(\rho(\rho\boxplus\rho)\bigr)
\right],
\qquad
\rho=\ket{\psi}\!\bra{\psi}.
\end{equation}
where $\boxplus$ denotes the fermionic convolution given by a balanced fermionic beam splitter
\begin{equation}
\label{eq:fermionic-bs-main}
\mathcal B_{\mathrm F}=\exp\!\left[\frac{\pi}{4}\sum_{j=1}^n(a_j^\dagger b_j-b_j^\dagger a_j)\right],
\end{equation}
with
\[
\mathcal B_{\mathrm F}^\dagger a_j\mathcal B_{\mathrm F}=\frac{a_j+b_j}{\sqrt2},
\qquad
\mathcal B_{\mathrm F}^\dagger b_j\mathcal B_{\mathrm F}=\frac{b_j-a_j}{\sqrt2}.
\]
$q_{\boxplus}(\psi)$ is exact for pure states where we have $q_{\boxplus}=0$ if and only if $\ket{\psi}$ is pure Gaussian~\cite{lyu2024convolution,coffman2025measuring}. 

This can be seen as a Bernoulli test with rejection probability
\[
q_{\boxplus}(\psi)=
\frac12\left(1-\operatorname{tr}[\rho(\rho\boxplus\rho)]\right).
\]
which requires \(N_\text{shot}=O(q_{\boxplus}^{-1}\log(1/\delta))\) shots.
A \(\epsilon\)-far testing bound would require a stability estimate encapsulated by function $s_n(\epsilon)$ with
\begin{equation}
    q_{\boxplus}(\psi)\ge s_n(\epsilon) \quad \forall \ket{\psi} \, \mathrm{with} \, D_{\rm tr}(\psi,\mathcal G_n)\ge\epsilon.
\end{equation}
as function of $n$, however no such bound is known to us.

\begin{figure}[htbp]
\centering
\includegraphics[width=0.33\textwidth]{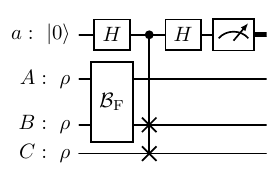}
\caption{Three-copy convolution test~\cite{lyu2024convolution,coffman2025measuring}. Two copies are mixed by the balanced fermionic beam splitter $\mathcal B_{\mathrm F}=\exp\!\left[\frac{\pi}{4}\sum_{j=1}^n(a_j^\dagger b_j-b_j^\dagger a_j)\right]$, which realizes $\rho\boxplus\rho$ after discarding the upper arm. A SWAP test with a fresh copy estimates $\Tr[\rho(\rho\boxplus\rho)]$, giving the test rejection probability $q_{\boxplus}(\rho)=\frac12(1-\Tr[\rho(\rho\boxplus\rho)])$.}
\label{fig:convolution-circuit-paper}
\end{figure}

\section{Higher-order FAF as a $2k$-copy observable}
\label{app:higher-faf}

We consider the Majorana bilinears
$
B_{ab}=-\ii\gamma_a\gamma_b,
\qquad a,b=1,\dots,2n,
$
with the convention $B_{aa}=0$, so that
$
\Tr(\rho B_{ab})=(\Gamma_\rho)_{ab}.
$
For $k\ge1$, define the $2k$-copy observable
\begin{equation}
\label{eq:higher-faf-observable}
\widehat T_k
=
\frac{(-1)^k}{2}
\sum_{a_1,\dots,a_{2k}=1}^{2n}
B_{a_1a_2}^{(1)}B_{a_2a_3}^{(2)}\cdots
B_{a_{2k}a_1}^{(2k)},
\end{equation}
where superscripts denote the copy on which an operator acts and indices are cyclically identified. Then
\begin{align}
\Tr(\widehat T_k\rho^{\otimes 2k})=
\frac{(-1)^k}{2}
\sum_{a_1,\dots,a_{2k}}
\Gamma_{a_1a_2}\Gamma_{a_2a_3}\cdots\Gamma_{a_{2k}a_1}=
\frac{1}{2}\Tr[(-\Gamma^2)^k].
\end{align}
The last equality follows from
$
\Tr(\Gamma^{2k})=
\sum_{a_1,\dots,a_{2k}}
\Gamma_{a_1a_2}\Gamma_{a_2a_3}\cdots\Gamma_{a_{2k}a_1}
$
and $(-\Gamma^2)^k=(-1)^k\Gamma^{2k}$. Therefore
\[
\FAF_k(\rho)
=
n-\Tr(\widehat T_k\rho^{\otimes 2k}).
\]
This proves the claimed $2k$-copy linearization of $\FAF_k$ (see also~\cite{faf2025}). Operationally, one may estimate \eqref{eq:higher-faf-observable} by sampling the indices $(a_1,\dots,a_{2k})$ and measuring one quadratic Majorana observable on each of the $2k$ copies, or by designing a collective commutant measurement. The $k=1$ observable is precisely the $\FAF_1$ observable.

\section{Wick-violation witnesses compared to FAF witness}
\label{app:wick-vs-faf}

We briefly review how Wick's theorem gives a direct way of witnessing fermionic non-Gaussianity~\cite{pachos2022quantifying,coffman2025measuring}, and then discuss a simple family of states for which low-order Wick tests are ineffective although the covariance--purity FAF witness directly detects non-Gaussianity.

For an ordered subset \(A=\{a_1<\cdots<a_\ell\}\subseteq [2n]\), let us write 
\[ \gamma_A = i^{\ell(\ell-1)/2} \gamma_{a_1}\cdots \gamma_{a_\ell}. \]
Then, Wick's theorem states that every even Majorana correlator is fixed by its covariance matrix:
\[
    \operatorname{tr}\!\left(\rho\,\gamma_{a_1}\cdots\gamma_{a_{2q}}\right)
    =
    \operatorname{Pf}\!\left(\Gamma_A\right),
    \qquad
    A=\{a_1,\ldots,a_{2q}\},
\]
where \(\Gamma_A\) is the \(2q\times 2q\) principal submatrix of \(\Gamma\) restricted to the indices in \(A\). Thus a Wick residual
\[
    \Delta_A(\rho)
    :=
    \operatorname{tr}(\rho\,\gamma_A)
    -
    \operatorname{Pf}(\Gamma_A)
\]
is a non-Gaussianity witness: if \(\Delta_A(\rho)\neq 0\) for some even subset \(A\), then \(\rho\) is not a fermionic Gaussian state.

The simplest nontrivial case is the four-point Wick residual. For four distinct indices \(i,j,k,\ell\), Wick's theorem predicts
\[
\begin{aligned}
    \operatorname{tr}(\rho\,\gamma_i\gamma_j\gamma_k\gamma_\ell)
    =
    \operatorname{tr}(\rho\,\gamma_i\gamma_j)
    \operatorname{tr}(\rho\,\gamma_k\gamma_\ell)
    -
    \operatorname{tr}(\rho\,\gamma_i\gamma_k)
    \operatorname{tr}(\rho\,\gamma_j\gamma_\ell)  
    +
    \operatorname{tr}(\rho\,\gamma_i\gamma_\ell)
    \operatorname{tr}(\rho\,\gamma_j\gamma_k),
\end{aligned}
\]
or, equivalently, in covariance-matrix notation,
\[
    \Delta_{ijkl}(\rho)
    =
    \operatorname{tr}(\rho\,\gamma_i\gamma_j\gamma_k\gamma_\ell)
    -
    \operatorname{Pf}(\Gamma_{\{i,j,k,\ell\}}).
\]
A nonzero \(\Delta_{ijkl}\) certifies non-Gaussianity. Checking all four-point residuals requires $\binom{2n}{4}$ correlators, which is polynomial in \(n\). More generally, checking all Wick identities up to order \(2K\) requires    $\sum_{q=2}^{K}\binom{2n}{2q}$
correlators, which is polynomial for fixed \(K\), but becomes exponential if one wants a complete Wick test over all even orders:
\[
    \sum_{q=2}^{n}\binom{2n}{2q}
    =
    2^{2n-1}-1-\binom{2n}{2}.
\]
Thus Wick witnesses are conceptually complete but experimentally expensive if no structure is assumed.

We now give a simple example illustrating this limitation. Define the fermionic parity operator
\[
    P
    =
    i^n \gamma_1\gamma_2\cdots\gamma_{2n},
\]
and consider the mixed state
\[
    \rho_\alpha
    =
    \frac{1}{2^n}
    \left(\mathbb{I}+\alpha P\right),
    \qquad
    0<|\alpha|\le 1.
    \label{eq:parity-biased-state}
\]
This is a valid density matrix because the eigenvalues of \(P\) are \(\pm1\), so the eigenvalues of \(\rho_\alpha\) are \((1\pm\alpha)/2^n\).

The state \(\rho_\alpha\) has a particularly simple Majorana moment structure. Since every non-identity Majorana monomial is traceless, and since \(P\gamma_A\) is again a non-identity Majorana monomial whenever \(0<|A|<2n\), we have
\[
    \operatorname{tr}(\rho_\alpha \gamma_A)=0,
    \qquad
    0<|A|<2n.
\]
For the Gaussian state with the same covariance matrix, namely the maximally mixed Gaussian state
\[
    \rho_{\mathrm{G}}
    =
    \frac{\id}{2^n},
\]
Wick's theorem predicts that every nonempty even correlator vanishes. Therefore \(\rho_\alpha\) satisfies every Wick identity at every order strictly below \(2n\), of which there are exponentially many. The only nonzero nontrivial Majorana moment is the full parity moment:
\[
    \operatorname{tr}(\rho_\alpha P)=\alpha,
\]
or equivalently
\[
    \operatorname{tr}\!\left(\rho_\alpha\,\gamma_1\gamma_2\cdots\gamma_{2n}\right)
    =
    (-i)^n\alpha.
\]
Since \(\operatorname{Pf}(\Gamma)=0\) for \(\Gamma=0\), the only independent Wick violation is the full \(2n\)-point residual
\[
    \Delta_{[2n]}(\rho_\alpha)
    =
    \operatorname{tr}\!\left(\rho_\alpha\,\gamma_1\gamma_2\cdots\gamma_{2n}\right)
    -
    \operatorname{Pf}(\Gamma)
    =
    (-i)^n\alpha.
\]
Hence this state is non-Gaussian, but its non-Gaussianity is hidden entirely in a top-order Majorana correlator.

This example is problematic for several natural Wick-based strategies.

First, measuring all Wick residuals up to a fixed order \(2K<2n\) fails deterministically. Indeed, all correlators of order below \(2n\) vanish and agree with the Wick prediction for \(\Gamma=0\). Thus any fixed-order Wick hierarchy is blind to the state \(\rho_\alpha\) once \(n>K\).

Second, random sampling of Wick residuals is also inefficient, as only one Wick-term is violated. As such, there is only an exponentially small probability in \(n\) to find a violation.

The parity-biased state is, however, detected immediately by the purity-corrected FAF witness. Since $\operatorname{tr}(\rho_\alpha^2) = \frac{1+\alpha^2}{2^n}$ and $\FAF_1(\rho_\alpha)=n$, we have
\begin{equation}
    W_{\text{FAF}}(\rho_\alpha)=n-2n\left[
    1-\frac{(1+\alpha^2)^{1/n}}{2}
    \right]
\end{equation}
Since \((1+\alpha^2)^{1/n}>1\) for every \(\alpha\neq 0\), we have $W_{\text{FAF}}(\rho_\alpha)>0$. Hence the purity-corrected FAF witness certifies that \(\rho_\alpha\) is not a mixed fermionic Gaussian state.

This example highlights the complementary nature of the two approaches. Wick witnesses are microscopic and complete if all even orders are measured: they identify the specific cumulants that violate Gaussian factorization. However, low-order Wick tests can miss non-Gaussianity that is stored in high-order correlations. The FAF witness is much more compressed. It does not identify the violating Wick residual, but it can certify that the covariance matrix and purity are incompatible with any mixed Gaussian state. In the example above, this compression is advantageous: the non-Gaussianity is invisible to all fixed-order Wick tests, yet it is seen directly from the FAF witness.

\end{document}